\documentclass[12pt,twoside]{article}
\usepackage{amssymb,amsmath,amsthm}
\usepackage[dvips]{graphicx}
\newcommand{\Var}{\mathop{\rm Var}}

\def\csect#1{Section~\ref{#1}}
\def\capp#1{Appendix~\ref{#1}}
\def\csects#1#2{Sections~\ref{#1} and \ref{#2}}

\let\Prp=\Pr \def\Pr{\Prp\nolimits}
\newcommand{\be}{\begin{equation}}
\newcommand{\ee}{\end{equation}}

\newtheorem{theorem}{Theorem}
\newtheorem{corollary}[theorem]{Corollary}
\newtheorem{lemma}[theorem]{Lemma}

\theoremstyle{remark}

\newtheorem{definition}[theorem]{Definition}
\newtheorem{remark}[theorem]{Remark}

\newcommand{\cthm}[1]{Theorem~\ref{#1}}

\newcommand{\clem}[1]{Lemma \ref{#1}}
\newcommand{\cthms}[3]{Theorems~\ref{#1}, \ref{#2}, and \ref{#3}}
\newcommand{\cdef}[1]{Definition~\ref{#1}}

\newcommand{\ccor}[1]{Corollary~\ref{#1}}

\newcommand{\cfig}[1]{Figure~\ref{#1}}
\newcommand{\crem}[1]{Remark~\ref{#1}}

\def\ds{\displaystyle}
\def\sss{\scriptscriptstyle}

    \def\bbz{\mathbb{Z}}

    \def\bbr{\mathbb{R}}
    \def\bbn{\mathbb{N}}

\def\n{{\bf n}}
\def\M{{\bf M}}

\def\muhat{\hat\mu}\def\rhohat{\hat\rho}
\def\murho{\mu^{(\rho)}}

\def\interiorM{\rlap{\raise9pt\hbox to9.5pt{\hss$\scriptscriptstyle\circ$}}M}

\def\M{\mathcal{M}}

\def\Xh{\bar X}
\def\Qh{\widehat Q}\def\Ph{\widehat P}

\def\Xz{X^{(0)}}

\def\ns{{\#}}

\def\Xl{X^{(l)}}\def\Xm{X^{(i)}}\def\Xh{X^{(h)}}
\def\Ul{U^{(l)}}\def\Ui{U^{(i)}}\def\Uh{U^{(h)}}\def\Us{U^{(\ns)}}
\def\Uhat{\hat U}

\long\def\kill#1\endkill{\relax}

\def\0{{\it0}}\def\1{{\it1}}\def\2{{\it2}}\def\3{{\it3}}
\def\rng#1#2{\hbox{$(#1\!:\!#2)$}}
\newdimen\mbsize \mbsize=\hsize \multiply\mbsize by 17  \divide\mbsize by 20

\numberwithin{equation}{section}
\numberwithin{theorem}{section}

\begin{document}

\title{\vskip-0.2truein
The Discrete-Time Facilitated Totally Asymmetric Simple Exclusion Process}
\author{S. Goldstein\footnote{Department of Mathematics,
Rutgers University, New Brunswick, NJ 08903.},\!
\footnote{Also Department of Physics, Rutgers.}\ \ 
J. L. Lebowitz\footnotemark[1],$\,$\footnotemark[2]\ \
and E. R. Speer\footnotemark[1]}
\date{\today}
\maketitle



\begin{abstract}We describe the translation invariant stationary states
of the one dimensional discrete-time facilitated totally asymmetric
simple exclusion process (F-TASEP).  In this system a particle at site
$j$ in $\bbz$ jumps, at integer times, to site $j+1$, provided site $j-1$
is occupied and site $j+1$ is empty.  This defines a deterministic
noninvertible dynamical evolution from any specified initial
configuration on $\{0,1\}^{\bbz}$.  When started with a Bernoulli product
measure at density $\rho$ the system approaches a stationary state, with
phase transitions at $\rho=1/2$ and $\rho=2/3$. We discuss various
properties of these states in the different density regimes $0<\rho<1/2$,
$1/2<\rho<2/3$, and $2/3<\rho<1$; for example, we show that the pair
correlation $g(j)=\langle\eta(i)\eta(i+j)\rangle$ satisfies, for all
$n\in\bbz$, $\sum_{j=kn+1}^{k(n+1)}g(j)=k\rho^2$, with $k=2$ when
$0\le\rho\le1/2$ and $k=3$ when $2/3\le\rho\le1$, and conjecture (on the
basis of simulations) that the same identity holds with $k=6$ when
$1/2\le\rho\le2/3$.  The $\rho<1/2$ stationary state referred to above is
also the stationary state for the deterministic discrete-time TASEP at
density $\rho$ (with Bernoulli initial state) or, after exchange of
particles and holes, at density $1-\rho$. \end{abstract}

\section{Introduction\label{intro}}

The {\it facilitated totally asymmetric simple exclusion process}
(F-TASEP) is a model of particles moving on the lattice $\bbz$, in which
a particle at site $j$ jumps, at integer times, to site $j+1$, provided
site $j-1$ is occupied and site $j+1$ is empty; many of the results on
this model appeared in \cite{glsshort}, without complete proofs.  The
related model in which the discrete time steps are replaced by a
continuous time evolution has been studied both numerically and
analytically \cite{bbcs,CZ,gkr}, as has the continuous-time model with
symmetric evolution \cite{BESS,oliveira}.  There are extensive numerical
simulations of similar models (usually called {\it Conserved Lattice
Gases}) in two or more dimensions \cite{hl,mcl,rpv}, but there are few
analytic results (but see \cite{ST}).  The model is of interest in part
because it exhibits nonequilibrium phase transitions.

A configuration of the model is an arrangement of particles on $\bbz$,
with each site either empty or occupied by a single particle; that is,
the configuration space is $X=\{0,1\}^\bbz$, with 1 denoting the presence
of a particle and 0 that of a hole.  We write $\eta=(\eta(i))_{i\in\bbz}$
for a typical configuration, and for $j,k\in\bbz$ with $j\le k$ we let
$\eta\rng jk=(\eta(i))_{j\le i\le k}$ denote the portion of the
configuration lying between sites $j$ and $k$ (inclusive). We will
occasionally use string notation, and correspondingly concatenation, for
configurations or partial configurations, writing for example
$\eta\rng06=\eta(0)\cdots\eta(6)=0\,1\,1\,0\,1\,0\,1=01^2(01)^2$.  
For $0\le\rho\le1$ we let $X_\rho\subset X$ denote the set of
configurations with a well-defined density $\rho$, that is,
configurations $\eta$ for which
 \be\label{rhodef}
\lim_{N\to\infty}\frac1N\sum_{i=1}^N\eta(i)
 = \lim_{N\to\infty}\frac1N\sum_{i=-N}^{-1}\eta(i) = \rho.
 \ee

Here we study the F-TASEP discrete-time dynamics as described above.  For
$\rho\notin\{1/2,2/3\}$ we will determine the ultimate fate of any
initial configuration $\eta\in X_\rho$.  We will also describe the
translation invariant (TI) states (i.e., TI probability measures on $X$)
of the system which are stationary under the dynamics (the TIS states);
without loss of generality we restrict consideration to states for which
almost all configurations have the same well-defined density $\rho$,
called {\it states of density $\rho$}, and will frequently assume further
that these states are ergodic under translations.  We would also like to
determine the final TIS state when the dynamics is started in a Bernoulli
measure: an initial state $\murho$ for which each site is independently
occupied with probability $\rho$.  In this, however, we will not be
completely successful.

We will make use of a closely related model, the {\it totally asymmetric
stack model} (TASM), another particle system on $\bbz$ evolving in
discrete time.  In the TASM there are no restrictions on the number of
particles at any site, so that the configuration space is
$Y=\bbz_+^\bbz$, where $\bbz_+=\{0,1,2,\ldots\}$.  We denote stack
configurations by boldface letters, so that a typical configuration is
$\n=(\n(i))_{i\in\bbz}$.  The dynamics is as follows: at each integer
time, every stack with at least two particles ($\n(i)\ge2$) sends one
particle to the neighboring site to its right.  This model is thus
essentially a discrete-time zero range process.

There is a natural correspondence between the TASM and the F-TASEP, with
a stack configuration $\n$ corresponding to a particle configuration in
which successive strings of $\n(i)$ particles are separated by single
holes; as just stated the correspondence is somewhat loose but yields a
bijective map $\psi:\Xz\to Y$, where $\Xz\subset X$ is the set of F-TASEP
configurations $\eta$ satisfying $\eta(0)=0$.  Moreover, if $\mu$ is a TI
probability measure on $X$ and we define
$\muhat=\mu(\Xz)^{-1}\mu\circ\psi^{-1}$ then $\mu\mapsto\muhat$ is a
bijective correspondence between TI, or TIS, probability measures on $Y$
and on $X$; this correspondence is discussed in detail in
\csect{correspondence}.  Using it, we show in \csect{phases} that there
are three phases for the F-TASEP, that is, three distinct regimes in
which the model exhibits qualitatively different behavior: the regions of
low, intermediate, and high density in which respectively $0<\rho<1/2$,
$1/2<\rho<2/3$, and $2/3<\rho<1$.

In subsequent sections we show, for each density region, how to determine
the ``final configuration'' resulting from the evolution of some
arbitrary initial configuration; we then suppose that the initial
configuration has a Bernoulli distribution and study the distribution of
the final configuration---that is, the TIS measure which is the
$t\to\infty$ limit of an initial Bernoulli measure (this problem was
studied for the continuous time model in \cite{CZ}).  It is for the low
density phase, treated in \csect{low}, that we can say the most.  We show
that every initial configuration $\eta_0$ of density $\rho<1/2$ has a
limit $\eta_\infty$---that is, it eventually freezes---and compute, for a
Bernoulli initial measure, the distribution of these final
configurations, which arises from a certain renewal process.  Moreover,
we show that if site $i$ is a point of this renewal process then the
expected density at any site an odd distance ahead of $i$ is $\rho$, and
that the two-point function in the final state,
$g(i)=\langle\eta_\infty(j)\eta_\infty(i+j)\rangle$, satisfies
$g(2n-1)+g(2n)=2\rho^2$ for any $n\ge1$; the latter property implies that
the asymptotic value of $V_L/L$, where $V_L$ is the variance of the
number of particles in an interval of length $L$, has the same value
$\rho(1-\rho)$ as for the initial Bernoulli measure.  We also compute the
distribution of the distance moved by a typical particle through the
evolution and find that the expected value of this distance is finite.
Finally, we show (see \crem{usual}) that the $\rho<1/2$ stationary state
is also the stationary state for the deterministic discrete-time totally
asymmetric simple exclusion process (TASEP) at density $\rho$ (in each
case with Bernoulli initial state) or, after exchange of particles and
holes, at density $1-\rho$.

A key technique for the study of the intermediate and high density
regions is to consider the dynamics in a moving frame; it is in this
frame that a limiting configuration exists for each initial
configuration.  Rather surprisingly, perhaps, the behavior of the model
in the high density region is largely parallel to that in the low density
region; we thus content ourselves with a rather brief treatment in
\csect{high}.  For the intermediate region, discussed in
\csect{intermediate}, the dynamics is considerably more complicated.
Here we are able to carry out the second step of the program, that is, to
determine the limit of the initial Bernoulli measure, only partially,
although we do show that the final measure can be characterized in terms
of a certain hidden Markov process.  Some technical and peripheral
results are relegated to appendices.

We mention finally some further observations about the model which can be
found in \cite{glsshort}.  When the empty lattice sites are regarded as
cars and the occupied sites as empty spaces, the model is closely related
to certain traffic models \cite{GG,LZGM}, with the low density region
corresponding to jammed traffic, the high density to free flow, and the
intermediate density to stop and go.  If in the low density phase an
initial Bernoulli measure is perturbed in some local way then the
perturbation does not dissipate; this is related to the finite expected
value of the distance a particle moves, mentioned above.  Finally, the
${\rm F}_k$-TASEP, defined by requiring that a particle have $k$ adjacent
particles to its left before it can jump, has properties analogous to the
F-TASEP itself; in particular, there are again three phases,
corresponding to density regions $\rho<k/(k+1)$,
$k/(k+1)<\rho<(k+1)/(k+2)$, and $(k+1)/(k+2)<\rho$ (the continuous-time
version of this model is discussed in \cite{BM}).

\section{Preliminary considerations}\label{preliminary}

We begin this section by introducing some notation to be used throughout
the paper.  We write $\bbz_\pm=\{0,\pm1,\pm2,\ldots\}$ and
$\bbn=\{1,2,\ldots\}$.  If $\lambda$ is a measure on a set $A$ and
$F:A\to\bbr$ then $\lambda(F)$ denotes the expectation of $F$ under
$\lambda$; if further $f:A\to B$ then $f_*\lambda$ is the measure on $B$
given by $(f_*\lambda)(C)=\lambda(f^{-1}(C))$.  Finally, we let $\tau$ be
the translation operator which acts on a function $f$ defined on $\bbz$
via $(\tau f)(k)=f(k-1)$.

\subsection{Correspondence of the F-TASEP and TASM}\label{correspondence}

In \csect{intro} we introduced the F-TASEP, with configuration space
$X=\{0,1\}^\bbz$, and the TASM, with configuration space $Y=\bbz_+^\bbz$;
in this section we establish the natural bijective correspondence between
the invariant measures for these two models.  This correspondence is
obtained from the substitution map $\phi:Y\to X$ defined by replacing
each $\n(i)$ in $\n=(\ldots,\n(-1),\n(0),\n(1),\ldots)$ by the string
$1^{\n(i)}0$, in such a way that the string for $\n(1)$ begins at site
$1$; thus for $\n\in Y$, $\eta=\phi(\n)$ has $\eta(i)=1$ for
$i=-\n(-1),\ldots,-1$, $\eta(0)=0$, $\eta(i)=1$ for $i=1,\ldots,\n(1)$,
$\eta(\n(1)+1)=0$, $\eta(i)=1$ for $i=\n(1)+2,\ldots,\n(1)+\n(2)+1$, etc.
Note that $\phi(Y)=\Xz:=\{\eta\mid\eta(0)=0\}$ and that $\phi^{-1}:\Xz\to Y$
is the map $\psi$ discussed in \csect{intro}.

We next show that $\phi$ gives rise to a bijection $\Phi$ from the space
of TI probability measures on $Y$ with finite density to the space of all
TI probability measures on $X$.  If $\muhat$ is a TI measure on $Y$,
$\muhat_n:=\muhat\big|_{Y_n}$ with $Y_n=\{\n\in Y\mid \n(1)=n\}$, and
$\mu_n=\phi_*\muhat_n$, then $\sum_{n\ge0}\sum_{i=0}^n\tau_*^{-i}\mu_n$
is a TI measure on $X$ of mass
$Z(\muhat):=\sum_{n\ge0}n\muhat(Y_n)=\muhat(\n(1))$.  If $Z(\muhat)$ is
finite we then define
 \be\label{defPhi}
\Phi(\muhat):=Z(\muhat)^{-1}\sum_{n\ge0}\sum_{i=0}^n\tau_*^{-i}\mu_n.
 \ee
 $\Phi(\muhat)$ is clearly TI and $\Phi$ is a bijection with inverse
$\Phi^{-1}:\mu\mapsto\muhat$ as described in \csect{intro}:
$\Phi^{-1}(\mu)=\mu(\Xz)^{-1}\psi_*\bigl(\mu\big|_{\Xz}\bigr)$.  $\Phi$
preserves convex combinations and this implies that $\muhat$ is ergodic
(i.e., extremal) if and only if $\Phi(\muhat)$ is.

To state our next result we let $U:X\to X$ and $\Uhat:Y\to Y$ be the
one-step evolution operators for the F-TASEP and TASM, respectively.

\begin{theorem}\label{TIS}(a) For any TI measure $\muhat$ on $Y$, with
finite density $Z(\muhat)$,
 \be\label{covariant}
  U^t_*\Phi(\muhat)=\Phi(\hat U^t_*\muhat).
 \ee

 \smallskip\noindent
 (b) $\Phi$ is a bijection of the TIS measures for the TASM and F-TASEP
systems.  \end{theorem}

\begin{proof} (b) is an immediate consequence of (a), and clearly it
suffices to verify (a) for $\muhat$ ergodic and $t=1$.  Let us write
$\nu:=U_*\Phi(\muhat)$ and $\widetilde\nu:=\Phi(\hat U_*\muhat)$.  Since
$U$ and $\Uhat$ preserve ergodicity, just as does $\Phi$, $\nu$ and
$\widetilde\nu$ are ergodic, so that these two measures are either equal
or mutually singular.  Hence to prove their equality it suffices to find
TI measures $\lambda$, $\lambda'$, and $\widetilde\lambda'$ on $X$, with
$\lambda$ nonzero, such that
 \be\label{decomp}
\nu=\lambda+\lambda'\text{ and }\widetilde\nu=\lambda+\widetilde\lambda'.
 \ee

The key identity relating the dynamics of the TASM and the F-TASEP,
easily checked, is that $U\phi(\n)=\tau^{-\gamma(\n(0))}\phi(\hat U\n)$,
where $\gamma(0)=\gamma(1)=0$ and $\gamma(n)=1$ if $n\ge2$.  Suppose now
that $n$ is such that $\muhat(\{\n\mid\n(0)=n\})>0$, and define
 \be
 \lambda=Z(\muhat)^{-1}U_*\phi_*\bigl(\muhat\big|_{\{\n(0)=n\}}\bigr), 
 \qquad
 \widetilde\lambda
   =Z(\muhat)^{-1}\phi_*\Uhat_*\bigl(\muhat\big|_{\{\n(0)=n\}}\bigr).
 \ee
 The identity given above implies that
$\lambda=\tau_*^{-\gamma(n)}\widetilde\lambda$, and it follows from
\eqref{defPhi} that $\nu-\lambda$ and $\widetilde\nu-\widetilde\lambda$
are (nonnegative) measures.  Then since $\widetilde\nu$ is TI,
 \be
\widetilde\nu=\tau_*^{-\gamma(n)}\widetilde\nu
  =\tau_*^{-\gamma(n)}(\widetilde\lambda+(\widetilde\nu-\widetilde\lambda))
  =\lambda+\tau_*^{-\gamma(n)}(\widetilde\nu-\widetilde\lambda);
 \ee
 this establishes \eqref{decomp}, with $\lambda'=\nu-\lambda$ and
$\widetilde\lambda'=\tau_*^{-\gamma(n)}(\widetilde\nu-\widetilde\lambda)$.
\end{proof}

 Note that if $\hat\mu$ is a TI state for the TASM, with density
$\hat\rho$ (in the sense that almost every configuration has density
$\hat\rho$, defined by the analogue of \eqref{rhodef}), then the
corresponding state $\mu$ of the F-TASEP has density
$\rho=\hat\rho/(1+\hat\rho)$.  If $\mu=\murho$ then in the corresponding
TASM measure $\muhat=\muhat^{(\rho)}$ the $\n(i)$ are i.i.d.~with
geometric distribution: $\muhat\{\n(i)=k\}=(1-\rho)\rho^k$.

\subsection{The three phases}\label{phases}

We begin with some simple observations on the dynamics in the TASM,
letting $\n_t(k)$ denote the height at time $t$ of the stack of particles
on site $k$.

\begin{itemize}
\item If $\n_t(k)\ge2$ then $\n_{t+1}(k)=\n_t(k)$ unless $\n_t(k-1)\le1$,
  in which case $\n_{t+1}(k)=\n_t(k)-1;$

\item If $\n_t(k)\le1$ then $\n_{t+1}(k)=\n_t(k)$ unless $\n_t(k-1)\ge2$,
  in which case $\n_{t+1}(k)=\n_t(k)+1.$
\end{itemize}

Thus the possible changes in the value of $\n(k)$ in one step of the
dynamics, say from $t$ to $t+1$, may be summarized as
 \be\label{shel}
0\,\rightarrow\,1\,\leftrightarrows\,2\,\leftarrow\,3\,\leftarrow
   \,4\,\leftarrow\,5\,\leftarrow\,\cdots.
 \ee
 The indicated increases occur if and only if $\n_t(k-1)\ge2$, and the
decreases if and only if $\n_t(k-1)\le1$.

 Suppose now that $\muhat$ is a TI state for the TASM.  For $n\in\bbz_+$
define the random variables $N_n$, $N_{\le n}$, and $N_{\ge n}$ on $Y$ by
 \be
 N_\ns(\n)=\lim_{N\to\infty}\frac1N\sum_{i=1}^N\chi_\ns(\n(i))
   \qquad (\text{$\ns=n$, $\le n$, or $\ge n$}),
 \ee
 where $\chi_n$, $\chi_{\le n}$, and $\chi_{\ge n}$ are the
characteristic functions of the sets $\{j\mid j=n\}$, $\{j\mid j\le n\}$,
and $\{j\mid j\ge n\}$, respectively.  (By the ergodic theorem, these
random variables are well-defined for $\muhat$-a.e.~$\n$).

\begin{lemma}\label{firststep} If $\muhat$ is a TIS state for the TASM
then for $\muhat$-a.e.~$\n$, (a)~either $N_{0}(\n)=0$ or
$N_{\ge2}(\n)=0$, and (b)~either $N_{\le1}(\n)=0$ or
$N_{\ge3}(\n)=0$.  \end{lemma}

\begin{proof} (a) Suppose to the contrary that there is a positive
probability that both $N_{0}>0$ and $N_{\ge2}>0$; then for some $k\ge1$,
which we may take to be minimal, $\muhat\{\n(0)\ge2,\n(k)=0\}>0$ (here we
have used the translation invariance of $\muhat$).  Now in fact
necessarily $k=1$, since minimality of $k$ implies that if $\n(0)\ge2$
and $\n(k)=0$ then $\n(1)=\cdots=\n(k-1)=1$, and if $k>1$ then at the
next time step we have $\n(1)=2$ and $\n(k)=0$, which by the stationarity
of $\muhat$ contradicts the minimality of $k$.  But if $\n(0)\ge2$ and
$\n(1)=0$ then at the next time step the empty stack at site 1
disappears; and since \eqref{shel} implies that empty stacks cannot be
created, this contradicts the stationarity of $\muhat$.

 \smallskip\noindent
 (b) Suppose that with positive probability both $N_{\ge3}>0$ (which by
(a) implies $N_0=0$) and $N_1>0$. Let $n$ be the minimal integer with
$n\ge3$ and $\muhat\{N_0>0,\,N_n>0\}>0$, and find as in (a) a minimal $k$
with
 \be
\muhat(\n(0)=1,\n(1)=\cdots=\n(k-1)=2,\n(k)=n)>0.
 \ee
 But then $k=1$, just as for (a), and  we again have a contradiction,
since when $\n(0)=1$ and $\n(1)=n$ the next time step yields $\n(1)=n-1$,
contradicting the minimality of $n$ or stationarity of $\muhat$.
\end{proof}

To state our next result we let $\Xl\subset X$ be the set of (low
density) configurations in which no two adjacent sites are occupied,
$\Xm\subset X$ be the set of (intermediate density) configurations in
which no two adjacent sites are empty and no three consecutive sites are
occupied, and $\Xh\subset X$ be the set of (high density) configurations
in which no two adjacent sites, and no two sites at a distance of 2 from
each other, are empty.

\begin{corollary}\label{threeregions} (a) Let $\muhat$ be a TIS state of
density $\rhohat$ for the TASM.  Then: if $\rhohat\le1$ then $N_{\ge2}=0$
$\muhat$-a.s.; if $1\le\rhohat\le2$ then $N_0=N_{\ge3}=0$ $\muhat$-a.s.;
and if $\rhohat\ge2$ then $N_{\le1}=0$ $\muhat$-a.s..
 \par\smallskip\noindent
 (b) Let $\mu$ be a TIS state of density $\rho$ for the TASM.  Then: if
$0\le\rho\le1/2$ then $\mu(\Xl)=1$ $\mu$-a.s.; if $1/2\le\rho\le2/3$ then
$\mu(\Xm)=1$ $\mu$-a.s.; and if $2/3\le\rho\le1$ then $\mu(\Xh)=1$
$\mu$-a.s..  \end{corollary}

\begin{proof}It is an immediate consequence of \clem{firststep} that the
three possibilities $N_{\ge2}(\n)=0$, $N_0(\n)=N_{\ge3}(\n)=0$, and
$N_{\le1}(\n)=0$ are exhaustive and mutually exclusive.  But these are
compatible only with $\rhohat\le1$, $1\le\rhohat\le2$, and $\rhohat\ge2$,
respectively, proving (a).  (b) is a direct translation of (a) from the
TASM language to the language of the F-TASEP.  \end{proof}

We will refer to the regions $0<\rho<1/2$, $1/2<\rho<2/3$, and
$2/3<\rho<1$ as the low, intermediate, and high density regions,
respectively (note the strict inequalities). \ccor{threeregions}
identifies $\Xl$, $\Xm$, and $\Xh$ as the supports of TIS measures in
these regions.  The supports take particularly simple forms at the
boundaries between regions: the support $\Xl\cap \Xm$ of a TIS measure
with $\rho=1/2$ consists of the two configurations in which 0's and 1's
alternate, so that the TIS measure, $\mu^*$, must assign weight 1/2 to
each of these configurations and is thus unique.  Similarly, there is a
unique TIS measure for $\rho=2/3$, which gives weight 1/3 to each of the
three configurations in $\Xm\cap \Xh$, that is, those with pattern
$\cdots\,0\,1\,1\,0\,1\,1\,0\,\cdots$.

The dynamics of the F-TASEP takes a simple form for configurations in
$\Xl$, $\Xm$, and $\Xh$: configurations in $\Xl$ do not change with time,
configurations in $\Xm$ translate two sites to the right at each time
step, and configurations in $\Xh$ translate one site to the left at each
time step.  (As an immediate consequence we see that {\it any} TI measure
on $\Xl\cup \Xm\cup \Xh$ is stationary.)  It is convenient then to
consider modified dynamics in the intermediate and high density regions,
under which the corresponding configurations are stationary.  In the low
density region we continue to use the original F-TASEP dynamics as
described in \csect{intro}; in the intermediate density region one first
executes, at each time step, the F-TASEP rule, then adds a translation by
two lattice sites to the left; in the high density region the evolution
is defined similarly, but the extra translation is by one site to the
right.  We introduce corresponding evolution operators $\Ul$, $\Ui$, and
$\Uh$, so that when discussing the evolution of an initial configuration
$\eta_0\in X_\rho$ with $\rho\notin\{0,1/2,2/3,1\}$ we will always write
$\eta_t=(\Us)^t\eta_0$ with $\ns=l$, $i$, or $h$ for $\rho$ in the low,
intermediate, or high density region, respectively.  Note that
$\Ui=\tau^{-2}\Ul$ and $\Uh=\tau\Ul$.

\begin{remark}\label{fixed} For the TASM, low or high density
configurations, i.e., those with $N_{\ge2}=0$ or $N_{\le1}=0$, are fixed
under the dynamics, while those of intermediate density, with
$N_0=N_{\ge3}=0$, translate one site to the right at each time step.
\end{remark}

As a final result of this section we show that TI measures always have
limits under the F-TASEP evolution.

\begin{theorem}\label{limit} Let $\mu_0$ be a TI measure and let
$\mu_t=U^t\mu_0$.  Then $\mu_\infty=\lim_{t\to\infty}\mu_t$ exists.
\end{theorem}

\begin{proof}We may assume without loss of generality that $\mu_0$ is
supported on $X_\rho$, $0\le\rho\le1$.  The result is trivial if $\rho=0$
or $\rho=1$.  For $\rho\notin\{0,1/2,2/3,1\}$ we prove below (see
\cthms{lowlimit}{highlimit}{midlimit}) that for any $\eta_0\in X_\rho$,
$\eta_\infty=\lim_{t\to\infty}\eta_t$ exists.  Thus with
$F:X_\rho\to X_\rho$ defined by $F(\eta_0)=\eta_\infty$ we have the
stated result, with $\mu_\infty=F_*\mu_0$, since $U^t\mu_0=(\Us)^t\mu_0$
because $\mu_0$ is TI.

We next suppose that $\rho=1/2$; the case $\rho=2/3$ is similar.  Let
$\delta_t=\mu_t(\eta(0)\eta(1))$ denote the densities of double 1's,
which must equal that of double 0's, at time $t$; it is easy to see that
$\delta_t$ is non-increasing in $t$, so that
$\delta_\infty=\lim_{t\to\infty}\delta_t$ exists.  As noted above, the
unique TIS measure at density $1/2$ is $\mu^*$; hence the
Ces\`aro means $t^{-1}\sum_{s=1}^t\mu_s$ converge to $\mu^*$ and this is
consistent only with $\delta_\infty=0$.  But then for any $L>0$ and any
$\epsilon>0$ there will be a $T$ such that for $t\ge T$ the marginal of
$\mu_t$ on $\{0,1\}^{[-L,L]}$ will, with probability at least
$1-\epsilon$, contain no double 1's or double 0's, and hence (using
translation invariance) coincide with the marginal of $\mu^*$. \end{proof}

\begin{remark}\label{fails}For $\rho=1/2$ (and similarly for $\rho=2/3$),
$\lim_{t\to\infty}\eta_t$ cannot exist for general $\eta_0\in X_\rho$.
For then as above we would have $\mu^*=\mu_\infty=F_*\mu_0$, where
$F(\eta_0)=\eta_\infty$, and if $\mu_0$ were the Bernoulli measure, then
because $F$ would commute with translations, $\mu^*$ would be
mixing, which it is not.
\end{remark}

\subsection{Height profiles}\label{height}

Suppose now that $\eta_t$ is a configuration of density
$\rho\notin\{1/2,2/3\}$ evolving by the dynamics above:
$\eta_{t+1}=\Us\eta_t$, with $\ns=l$, $i$, or $h$.  We define a
corresponding {\it height profile} $h_t:\bbz\to\bbz$ which, in the usual
convention, rises by one unit when $\eta_t(i)=0$ and falls by one unit
when $\eta_t(i)=1$:
 \be\label{eta-h}
h_t(k)-h_t(k-1)=(-1)^{\eta_t(k)}.
 \ee
  Now \eqref{eta-h} defines $h_t$ only up to an additive constant; to
specify this we first define the initial profile by making the
arbitrary choice $h_0(0)=0$, which with \eqref{eta-h} leads to
 \be\label{profdef}
  h_0(k)=\begin{cases}\openup4\jot 0,&\hbox{if $k=0$,}\\
   \sum_{i=1}^k(-1)^{\eta_0(i)},& \hbox{if $k>0$,}\\
   -\sum_{i=k+1}^{0}(-1)^{\eta_0(i)},& \hbox{if $k<0$.}\end{cases} 
 \ee
 Next we want to define the evolution operator on profiles, again denoted
$\Us$, choosing the additive constant at each step so that $h_t$ is
stationary when $\eta_t$ is.  For $\rho$ in the low density region this
means that, given $h_t$ (and $\eta_t$, which may be obtained from $h_t$
via \eqref{eta-h}), we take $h_{t+1}(k)=(\Ul h_t)(k)$, with
 \be\label{profmove}
 (\Ul h_t)(k)=\begin{cases}
   h_t(k)+2,&\text{if $\eta_t(k-1)=\eta_t(k)=1$ and $\eta_t(k+1)=0$},\\
     h_t(k),& \text{otherwise.}\end{cases} 
 \ee
For $\rho$ in the intermediate density region we must include a
translation: $h_{t+1}=\Ui h_t=\tau^{-2}\Ul h_t$, and in the high
density region we need also a vertical shift:
$h_{t+1}=\Uh h_t=\tau \Ul h_t-1$.

In the intermediate and high density regions we will use also a modified
height profile: $h_t^*(k)=h_t(k)+k/3$, $k\in\bbz$.  It is easy to verify,
using $h_0(0)=0$, that these profiles satisfy
 \be\begin{gathered}\label{highmod}
{\rm(i)}\quad k+h_t(k)\equiv0\bmod2,\qquad
 {\rm(ii)}\quad3h^*_t(k)\equiv0\bmod2,\\
 {\rm(iii)}\quad k+\frac32 h^*_t(k)\equiv0\bmod3.
 \end{gathered}\ee

\section{The low density region}\label{low}

In this section we study the dynamics in the low density region
$0<\rho<1/2$.  A key role will be played by the height profile $h$
of \csect{height}.

\subsection{Evolution of a single configuration}\label{lowfate}

Here we fix an initial configuration $\eta_0$ with density $\rho$,
$0<\rho<1/2$, and let $\eta_t$ and $h_t$ be the corresponding evolving
configuration and height profile (as in \csects{phases}{height}).  We
define a subset $P=P(\eta_t)\subset\bbz$ by
 \be\label{lowpk} p\in P \quad\text{iff}\quad h_t(p)> \sup_{i<p}h_t(i).
 \ee
  (\cthm{lowlimit}(a) below justifies our suppression in \eqref{lowpk} of
the apparent $t$ dependence of $P$.) \eqref{rhodef} implies that $h_t$
has mean slope $1-2\rho>0$, so that
$\lim_{i\to\pm\infty}h_t(i)=\pm\infty$ and hence $P$ is unbounded above
and below.  Note further that as $p$ runs over $P$, $h_t(p)$ takes each
value in $\bbz$ precisely once, that if $p$ and $p'$ are consecutive
elements of $P$ then $h_t(p')=h_t(p)+1$, and that if $p\in P$ then
$\eta_t(p-1)=\eta_t(p)=0$.  It follows from \eqref{lowetainf} below that
$P$ is precisely the set of points $p$ with
$\eta_\infty(p-1)=\eta_\infty(p)=0$.

\begin{theorem}\label{lowlimit} (a) $P$ as defined in \eqref{lowpk} is
independent of $t$.
 
\par \smallskip\noindent
 (b) For each $i\in\bbz$, $\eta_t(i)$ and $h_t(i)$ are nondecreasing in
$t$ and eventually constant.  If we denote these limiting values by
$\eta_\infty(i)$ and $h_\infty(i)$ then for $p$ and $p'$ any consecutive
points of $P$ there is an $n$ with
 \be\label{lowetainf}
\eta_\infty\rng{p+1}{p'}=1\,0\,1\,0\,\cdots\,1\,0\,0=(1\,0)^{n}0. 
 \ee
\end{theorem}

We note that \eqref{lowetainf} specifies $\eta_\infty$ completely.  A
graphical representation of the contents of this theorem is shown in
\cfig{fig:ldprof}, where the profiles $h_0$ and $h_\infty$ are
represented as piecewise linear curves in the plane which are obtained by
connecting each pair of points $(i,h(i))$ and $(i+1,h(i+1))$ by a
straight line segment.

\begin{figure}[ht!]  \begin{center}
\includegraphics[width=5.0truein,height=2.0truein]{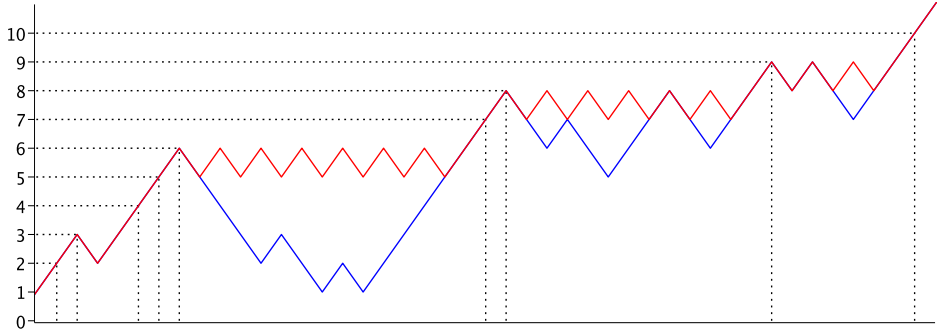} 
\caption{Portion of typical initial (blue, lower) and final (red, upper) height
profiles in the F-TASEP.  The vertical dotted lines are at sites in $P$.}
\label{fig:ldprof} 
\end{center}
\end{figure}

  \begin{proof}[Proof of \cthm{lowlimit}] For the moment we denote the
set defined in \eqref{lowpk} by $P_t$.  Observe first that
\eqref{profmove} implies that for fixed $i$, $h_t(i)$ is nondecreasing in
$t$; moreover, $h_{t+1}(i)>h_t(i)$ is possible only if $h_t(i-1)>h_t(i)$.
This implies that if $p\in P_0$ then for all $t\ge0$, $h_t(p)=h_0(p)$ and
$h_t(p)>\max_{i<p}h_t(i)$.  Thus $p\in P_t$ and so $P\subset P_t$; since
$h_t(p)$ takes each value in $\bbz$ precisely once, $P_t=P$, verifying
(a).  Moreover, for any $i\in\bbz$ there will be a $p\in P$ with $p>i$,
and the upper bound $h_t(i)<h_t(p)=h_0(p)$ shows the existence of the
limit $h_\infty(i)=\lim_{t\to\infty}h_t(i)$, and hence, via
\eqref{eta-h}, also of the limit $\eta_\infty(i)$.  Further, if $p$ and
$p'$ are two consecutive elements of $P$ and $p\le i<p'$ then
$h_\infty(i)\ge h_\infty(p)-1$, since if $h_\infty(i)\le h_\infty(p)-2$
for some $i$ with $i>p$ then necessarily
$\eta_\infty\rng j{j+2}=1\,1\,0$ for some $j$ with $p<j<p'$, and an
exchange must then take place, contradicting the time-independence of
$\eta_\infty$.  The conclusion that
$h_\infty(p)-1\le h_\infty(i)\le h_\infty(p)$ for $p\le i<p'$ yields
\eqref{lowetainf}.\end{proof}

\subsection{A Bernoulli initial distribution}\label{lowbern}

In this section we assume that the initial configuration $\eta_0$ is
distributed according to the Bernoulli measure $\murho$, with
$0<\rho<1/2$.  Then almost every initial configuration $\eta_0$ satisfies
\eqref{rhodef}; for such configurations the set $P$ of \eqref{lowpk} is
well defined, the analysis of the preceding section applies, and
$\eta_\infty$ is determined as a function of $\eta_0$.  \cthm{lowlimit}
suggests that to obtain the distribution of $\eta_\infty$ we should
obtain the joint distribution of the (ill-defined at the moment) ``random
variables'' $p'-p$ of \eqref{lowetainf}.  To state a precise result we
would like to index the points of $P$, with $p_{k}<p_{k+1}$ for all $k$,
but unfortunately this cannot be done without introducing some unwanted
bias into the differences $p_{k+1}-p_k$.

To deal with this problem we first introduce the set
$V:=\{\eta_0\mid0\in P(\eta_0)\}$ and let $\mu$ denote the measure
$\murho$ conditioned on $V$ (we could just as well replace $V$ by
$\{\eta_0\mid j\in P(\eta_0)\}$ for any $j\in\bbz$).  For configurations
$\eta_0\in V$ we label the points of $P(\eta_0)$ so that $p_0=0$ and
$p_k<p_{k+1}$. To describe the distribution of the differences
$p_{k+1}-p_k$ under $\mu$ we will use the {\it Catalan numbers}
 \be\label{catalan}
   c_n=\frac1{n+1}\binom{2n}n,\quad n=0,1,2,\ldots
 \ee
 (the sequence is entry {\bf A000108} in the Online Encyclopedia of
Integer Sequences \cite{OEIS}).  $c_n$ counts the number of strings of
$n$ 0's and $n$ 1's in which the number of 0's in any initial segment
does not exceed the number of 1's, or alternatively the number of {\it
Dyck paths} of length $2n$: paths in the lower half plane, with possible
steps $(1,1)$ and $(1,-1)$, from $(0,0)$ to $(2n,0)$.

\begin{theorem}\label{ldrenewal}The random variables $N_k=p_{k+1}-p_k$
are i.i.d.~under $\mu$, with distribution
 \be\label{catdist}
\mu(\{N_k=2n+1\})=c_n\rho^n(1-\rho)^{n+1}, \quad n=0,1,2,\ldots,
 \ee
 \end{theorem}

Note that the independence of the variables $N_k$ implies that the set
$P$ is a renewal point process; this is the renewal process mentioned in
\csect{intro}.

\begin{proof}[Proof of \cthm{ldrenewal}]Whether or not a site $p$ belongs
to $P$ is determined by the $\eta(i)$ with $i\le p$, whereas given that
$p\in P$, the next point $p'$ of $P$ is determined by the $\eta(i)$ with
$i>p$.  This establishes the independence of the increments $N_k$.  More
specifically, if $p\in P$ and $l=p+2n+1$ then $p'=l$ if and only if
$h_0(l)=h_0(p)+1$ and $h_0(i)\le h_0(p)$ for $p<i<l$; the latter
condition holds if and only if in the string $\eta\rng{p+1}{l-1}$ the
number of 0's in any initial segment does not exceed the number of 1's,
that is, if the segment of $h_0(i)$ for $p\le i<p+l$ forms a Dyck path.
Thus there are $c_n$ configurations of $\eta\rng{p+1}l$ yielding $p'=l$,
and since each such configuration has probability $\rho^n(1-\rho)^{n+1}$,
\eqref{catdist} is established. \end{proof}

\begin{remark}\label{muinf}(a) The distribution $\mu_\infty$ of
$\eta_\infty$ may be expressed in terms of $\mu$ by a standard
construction:
$\mu_\infty=Z^{-1}\sum_{m\ge0}\sum_{i=0}^{m-1}\tau_*^{-i}
\bigl(\mu\big|_{V_m}\bigr)$, where $V_m=\{\eta_0\in V\mid
p_1-p_0=m\}$ and $Z:=\sum_{m\ge0}m\mu(V_m)$ (compare the construction of
$\Phi$ in \csect{correspondence}).

 \smallskip\noindent
 (b) In the {\it continuous-time Facilitated Partially Asymmetric
Exclusion Process} the transitions $1\,1\,0\to1\,0\,1$ and
$0\,1\,1\to1\,0\,1$ occur at rates $p$ and $1-p$, respectively.  It can
be shown \cite{agls} that if this process is started in a Bernoulli
measure with density $\rho<1/2$ then the final state is again described
by the measure $\mu_\infty$ of (a), whatever the value of
$p$.\end{remark}

We now discuss some further properties of the final state of the system,
still when started from a Bernoulli measure.  

\begin{lemma}\label{density}For any $i\in\bbz$,
$\murho(\{i\in P(\eta_0)\})=1-2\rho$.  \end{lemma}

\begin{proof}One can verify the result by direct consideration of the
initial state, but it is easier to observe from \cthm{lowlimit}(b) that the
desired probability is just
$\murho(\{\eta_\infty(i-1)=\eta_\infty(i)=0\})$, and the result then
follows from $\murho(\eta_\infty(i-1))=\murho(\eta_\infty(i))=\rho$ and
$\murho(\eta_\infty(i-1)\eta_\infty(i))=0$.  \end{proof}

Recall now the definitions of $V$ and $\mu$ given above; note that
$\murho(V)=1-2\rho$ and that, if $\eta_0\in V$, then during the evolution
of $\eta_0$ no particle can cross the bond $\langle0,1\rangle$.  This
implies that $\eta_\infty\rng1\infty$ depends only on
$\eta_0\rng1\infty$, and in such a manner that $\eta_\infty(1)=\eta_0(1)$
so that $\mu(\eta_\infty(1))=\murho(\eta_0(1))=\rho$. 

\begin{lemma}\label{addx}For any $n\ge1$,
$\mu(\{2n\in P\})=\mu(\{2n+1\in P\})$.  \end{lemma}

\begin{proof}For $\eta_0\in V$ and $x\in\{0,1\}$ let $\eta_0^x$ be
obtained by inserting $x$ into $\eta_0$ immediately to the right of the
origin: $\eta_0^x(i)=\eta_0(i)$ for $i\le0$, $\eta_0^x(1)=x$, and
$\eta_0^x(i)=\eta_0(i-1)$ for $i\ge2$.   We claim that
$2n\in P(\eta_0)$ iff $2n+1\in P(\eta_0^0)\cap P(\eta_0^1)$,
which immediately implies the result.  For the claim, 
let $h_0$ be the height function
for $\eta_0$ and $h_0^x$ that for $\eta_0^x$.  Note first that if
$2n+1\in P(\eta_0^1)$ then $2n\in P(\eta_0)$; this is clear
geometrically, since in passing from $h_0^1$ to $h_0$ we raise, and shift
one site to the left, the portion of the height profile to the right of
site 1.  (An analytic proof similar to the argument given just below is
easy to write down.)  Next, if $2n\in P(\eta_0)$ then similarly
$2n+1\in P(\eta_0^0)$.  Finally, we check in more detail that if
$2n\in P(\eta_0)$ then $2n+1\in P(\eta_0^1)$.  We know that
$h_0(2n)>h_0(i)$ for any $i<2n$ and must show that $h_0^1(2n+1)>h_0^1(i)$
for $i<2n+1$.  Now if $i\ge1$ this follows from
$h_0^1(2n+1)=h_0(2n)-1>h_0(i-1)-1=h_0^1(i)$, while if $i\le0$ we use the
fact that $h_0(2n)$ is even and $h_0(2n)>h_0(0)=0$ to write
$h_0^1(2n+1)=h_0(2n)-1\ge1>h_0(0)\ge h_0(i)=h_0^1(i)$.  \end{proof}

In stating the next theorem we let $g$ denote the two-point correlation
function in the final state: $g(k)=\murho(\eta_\infty(0)\eta_\infty(k))$.

\begin{theorem}\label{twopt}(a) For any $n\ge1$,
  $\mu(\eta_\infty(2n-1))=\rho$.
 \par\smallskip\noindent
 (b) For any $n\ge1$, $g(2n-1)+g(2n)=2\rho^2$.
\end{theorem}

\begin{proof}From \clem{addx} and the fact that
$\murho(\eta_\infty(i)\eta_\infty(i+1))=0$ for any $i$ it follows that
the distribution under $\mu$ of
$(\eta_\infty(2n-1),\eta_\infty(2n),\eta_\infty(2n+1))$ is symmetric
under the exchange of the first and last variables.   Thus
$\mu(\eta_\infty(2n-1))=\mu(\eta_\infty(2n+1))$, and this, with the
observation above that $\mu(\eta_\infty(1))=\rho$, yields (a).  But then
(b) follows immediately, from
 \begin{align*}
   \mu(\eta_\infty(2n-1))
  &=\frac1{1-2\rho}\murho((1-\eta_\infty(-1))
    (1-\eta_\infty(0))\eta_\infty(2n-1))\\
  &=\frac1{1-2\rho}(\rho-g(2n)-g(2n-1)).\qedhere
 \end{align*}
\end{proof}

An alternative proof of \cthm{twopt}(a)---which, in fact, generalizes
that result to $\mu(\eta_t(2n-1))=\rho$ for all $t$---is presented in
\capp{semiinf}.  \cthm{twopt}(b) is then a consequence, as above. A third
proof of the latter is obtained from the computation in \capp{genfct} of
the generating function for $g(k)$.

We next observe that the truncated two point function
$g^T(k):=g(k)-\rho^2$ decays exponentially.

\begin{lemma}\label{decay}Let $\alpha_0:=(4\rho(1-\rho))^{1/2}<1$.  Then
for any $\alpha>\alpha_0$ there is a $C_\alpha>0$ such that
$|g^T(k)|\le C_\alpha\alpha^{k}$.  \end{lemma}

\begin{proof}[Proof sketch]One finds the generating function
$G^T(z):=\sum_{n=1}^\infty g^T(n)z^n$ and observes that it is
analytic for $|z|<1/\alpha_0$.  Some details are given in
\capp{genfct}.  \end{proof}

We next consider the variance of the number of particles in large boxes.

\begin{theorem}\label{variance}Let $S_L=\sum_{i=1}^L\eta_\infty(i)$ and
$V_L=\Var(S_L)=\murho(S_L^2)-(\rho L)^2$. Then for $0<\rho<1/2$,
$\lim_{L\to\infty} V_L/L = \rho(1-\rho)$.  \end{theorem}

\begin{proof}The exponential decay of \clem{decay} is amply sufficient to
  justify the standard formula 
 \be
\lim_{L\to\infty}\frac{V_L}L=\rho(1-\rho)+\sum_{k=1}^\infty g^T(k).
 \ee
 But by \cthm{twopt}, $\sum_{k=1}^n g^T(k)=0$ if $n$ is even. \end{proof}

The result of \cthm{variance} may also be understood in terms of the fact
that, as we next discuss, a typical particle moves only a microscopic
distance during the evolution.  Thus the number of particles in a large
box is, to high relative accuracy, the same at the end of the evolution
as it was at the beginning.  We will in fact show in \cthm{distance} that
the distance moved by such a typical particle has a geometric
distribution with mean $\rho/(1-2\rho)$.

 Consider then a particle initially located at a site $i$, with
$p_k<i<p_{k+1}-1$, and let $i_t$ be its position at time $t$.  During the
evolution, $h_t(i_t)$ will increase from $h_0(i)$ to $h_0(p_k)-1$, so
that the particle will move a distance $h_0(p_k)-h_0(i)-1$.  Now consider
further the collection of all Dyck paths of length $2n$, which for the
moment we think of as starting at $(0,0)$; there are $c_n$ such paths and
each configuration described by one of them contains $n$ particles, for a
total of $nc_n$ particles.  Let $\Delta(n,d)$ be the number of these
particles which will move a distance exactly $d$, and note that
$\Delta(0,d)=0$.  By conditioning on the site $2m$ where the path first
returns to height 0 (a standard trick for obtaining the recursion for
Catalan numbers) we find the recursion
 \be\label{recursion}
\Delta(n,d)=\sum_{m=1}^n\bigl(c_{n-m}\Delta(m-1,d-1)
    +c_{m-1}\Delta(n-m,d)\bigr).
 \ee
 This relation holds even for $d=0$ if we define $\Delta(n,-1)=c_n$.

 Now introduce the generating functions $G(u)=\sum_{n=0}^\infty c_nu^n$
and $G_d(u)=\sum_{n=0}^\infty \Delta(n,d)u^n$ (so that $G(u)=G_{-1}(u)$).
It is well-known \cite{Catalan} that $G(u)$ satisfies $G(u)=1+uG(u)^2$
and that explicitly $G(u)=2/(1+\sqrt{1-4u})$.  From \eqref{recursion} we
have $G_d(u)=uG(u)(G_{d-1}(u)+G_d(u))$, easily solved to give
 \be
G_d(u)=G(u)(G(u)-1)^{d+1}.
 \ee

 We next condition on there being a particle at the origin in $\eta_0$,
let $D$ be the distance that that particle moves, and find the
distribution of $D$.  For some $k$ we will have $p_k<0<p_{k+1}$; we first
calculate the probability $\pi_n$ that $p_{k+1}-p_k=2n+1$.  In that event
there are $n$ possible sites for $p_k$, the probability that a selected
site lies in $P(\eta_0)$ is $(1-2\rho)$ (\clem{density}), and we must
divide by $\rho$ to condition on $\eta_0(0)=1$, so that from
\eqref{catdist},
 \be
\pi_n = n(1-2\rho)c_n\rho^{n-1}(1-\rho)^{n+1}.
 \ee
 But since, given that $p_{k+1}-p_k=2n+1$, all compatible Dyck paths and
 positions of the origin relative to the path are equally likely, 
 \begin{align}
\murho(D=d\mid\eta_0=1)&=\sum_{n=1}^\infty\pi_n\frac{\Delta(n,d)}{nc_n}\nonumber\\
 &=\frac{(1-\rho)(1-2\rho)}{\rho}G_d(\rho(1-\rho))\nonumber\\
 &= \frac{1-2\rho}{1-\rho}\left(\frac\rho{1-\rho}\right)^d.
 \end{align}
 We have proved:

\begin{theorem}\label{distance}The distance $D$ moved by a ``typical''
particle, i.e., by the particle at the origin given that at time zero
there is such a particle, has geometric distribution, with ratio
$\rho/(1-\rho)$ and mean $\rho/(1-2\rho)$.  \end{theorem}
 
\begin{remark}\label{usual} Consider the deterministic discrete-time
TASEP, in which all particles with an empty site to their left jump to
that site at integer times.  (We have reversed the conventional choice of
jump direction, with which the model is also called CA 184 \cite{BF}, for
reasons to be seen shortly.) The model is often studied in a
probabilistic version, in which each jump takes place with some
probability $p$; in this case there is a unique TIS state \cite{Evans}.
For the deterministic model with density $\rho<1/2$, however, {\it any}
TI state in which, with probability 1, each particle is isolated, is
stationary, since each configuration simply translates to the left with
velocity 1.  These are all the TIS measures \cite{BF}.  It is then
natural to consider a modified dynamics in which, at each time step, one
first does a TASEP update, then translates all particles to the right by
one site.  This gives a modification of the facilitated dynamics: an
isolated particle does not move, but if there is a block of $k$ particles
then the left-most one stays fixed and the remaining $k-1$ move one step to
the right.  Our analysis of the F-TASEP through the height function can
then be applied directly, so that an initial configuration $\eta_0$
evolves under the modified TASEP dynamics to the same $\eta_\infty$ as in
the F-TASEP.

For TI initial (and hence final) states the modification of the dynamics
will not affect stationarity, so that if we start the system in some TI
measure $\lambda$ then the final measure will be the F-TASEP final
measure $\lambda_\infty$; in particular, if $\lambda$ is Bernoulli then
the final measure will be the one of \crem{muinf}.  On the other hand, if
$\lambda$ is the initial measure for the TASEP in which particles move to
the right then the final measure will be $R_*((R_*\mu)_\infty)$, where
$R:X\to X$ is reflection.  For either direction of motion, stationary
states at $\rho>1/2$ may then be determined through the usual
particle-hole symmetry.  \end{remark}

 \section{The high density region}\label{high}

We now turn to the high density region $2/3<\rho<1$; we will be brief,
because the behavior of the model here is very similar to that in the low
density region.  Recall from \csects{preliminary}{height} that
the dynamics will now be given by $\eta_{t+1}=\Uh\eta_t$ and
$h_{t+1}=\Uh h_t$. 

Let us first fix an initial configuration $\eta_0$ with density $\rho$,
$2/3<\rho<1$, and determine its final form $\eta_\infty$.  We define
$Q:=\{q\in\bbz\mid h_t^*(q)<\inf_{i<q}h_t^*(i)\}$; \cthm{highlimit}(a)
below justifies this notation, which ignores the apparent $t$ dependence
of $Q$.  Since $h_t$ has mean slope $1-2\rho<-1/3$,
$\lim_{i\to\pm\infty}h^*_t(i)=\mp\infty$, so that $Q$ is well defined and
unbounded above and below.  Note that if $q\in Q$ then
$\eta\rng{q-2}{q}=1\,1\,1$ for all $k$.

\begin{theorem}\label{highlimit} (a) $Q$ as defined above is independent
of $t$.
 
\par \smallskip\noindent
 (b) For each $i\in\bbz$, $\eta_t(i)$ and $h_t(i)$ are eventually
constant, and if we denote these limiting values by $\eta_\infty(i)$ and
$h_\infty(i)$, then  then for $q$ and $q'$ any consecutive
points of $Q$ there is an $n$ with
 \be\label{highetainf}
\eta_\infty\rng{q+1}{q'}=(0\,1\,1)^nd\,1. 
 \ee
\end{theorem}

\begin{proof} The proof is completely parallel to that of \cthm{lowfate}.
One first checks that, for $i\in\bbz$, $h^*_t(i)$ is nonincreasing in
$t$, with $h^*_{t+1}(i)<h^*_t(i)$ possible only if $h^*_t(k-1)>h^*_t(k)$;
this implies that $Q$ is time-independent and that the $h^*_\infty$, and
hence $h_\infty$ and $\eta_\infty$, exist.  Then the fact that
$\eta_\infty\in\Xh$ (see \ccor{threeregions}(c)) yields
\eqref{highetainf}.  \end{proof}

Now we suppose that the initial configuration $\eta_0$ is distributed
according to the Bernoulli measure $\murho$, with $\rho$ satisfying
$2/3<\rho<1$.  Let $V^*:=\{\eta_0\mid0\in Q(\eta_0)\}$, let $\mu^*$ denote
the measure $\murho$ conditioned on $V^*$, and for $\eta_0\in Q$ index
the points of $Q$ in increasing order, with $q_0=0$.  We first find the
distribution of the differences $q_k-q_{k-1}$, and in doing so will refer
to the sequence
 \be\label{raney}
d_n=\frac1{2n+1}\binom{3n}{n},\qquad n=0,1,2,\ldots.
 \ee
 The $d_n$ are a particular case of {\it Fuss-Catalan} or {\it Raney}
numbers \cite{MPZ} (OEIS entry {\bf A001764} \cite{OEIS}).  $d_n$ counts
the number of paths from the origin to $(3n,-n)$, with possible steps
$(1,1)$ and $(1,-1)$, such that the path never goes below the line
$y=-x/3$.

\begin{theorem}\label{hdrenewal}The random variables $N^*_k=q_{k+1}-q_k$
are i.i.d. under $\mu^*$, with distribution
 \be\label{randist}
\mu^*\{N^*_k=3n+1\}=d_n\rho^{2n+1}(1-\rho)^n, \quad n=0,1,2,\ldots,
 \ee
 \end{theorem}

\begin{proof} Independence of the increments $N^*_k$ is established as in
the proof of \cthm{ldrenewal}.  As in that proof, if $l=q_k+3n+1$ then
$q_{k+1}=l$ if and only if $h^*_0(l)=h^*_0(q_k)-2/3$ and the segment of
$h_0(i)$ for $q_k\le i<q_k+l$ is a path of the type counted by $d_n$ (see
the previous paragraph).  Thus there are $d_n$ configurations of
$\eta\rng{q_k+1}l$ yielding $q_{k+1}=l$, each with probability
$\rho^{2n+1}(1-\rho)^n$, establishing \eqref{catdist}.  \end{proof}

Our next theorem summarizes various results for the high density region
which are parallel to the results of \csect{low} for the low density
region.  In stating these we define, as in \csect{low},
$g(k)=\murho(\eta_\infty(0)\eta_\infty(k))$.  We omit the proofs, all of
which are modifications of those of the previous section.

\begin{theorem}\label{hdvarious}Suppose that $2/3<\rho<1$.  Then:

 \smallskip\noindent
 (a) For any $i\in\bbz$, $\murho(\{i\in Q(\eta_0)\})=3\rho-2$.

 \smallskip\noindent
 (b) For any $n\ge1$, $\mu^*(\{3n\in Q\})=\mu^*(\{3n+1\in Q\})$.

 \smallskip\noindent
 (c) For any $n\ge0$, $\mu^*(\{\eta_\infty(3n+1)=1\})=\rho$.
 
\smallskip\noindent
 (d) For any $n\ge1$, $g(3n-2)+g(3n-1)+g(3n)=3\rho^2$.

 \smallskip\noindent
 (e) Let $\alpha^*_0:=(27\rho^2(1-\rho)/4d)^{1/3}<1$. Then for any
$\alpha^*>\alpha^*_0$ there is a $C_{\alpha^*}>0$ such that
$|g^T(k)|\le C_{\alpha^*}\alpha^{*k}$.

 \smallskip\noindent
 (f) Let $S_L=\sum_{i=1}^L\eta_\infty(i)$ and
$V_L=\Var(S_L)=\murho(S_L^2)-(\rho L)^2$. Then for $2/3<\rho<1$,
$\lim_{L\to\infty} V_L/L = \rho(1-\rho)$.  \end{theorem}

There is no analogue of \cthm{distance}, since in the high density region
particles never stop moving, either in the original dynamics given by $U$
or the modified dynamics given by $\Uh$.

\section{The intermediate density region}\label{intermediate}

In this section we study the dynamics in the intermediate density region
$1/2<\rho<2/3$.

\subsection{Evolution of a single  configuration}\label{interfate}

 We first discuss the evolution of a given configuration of density
$\rho$, beginning with some preliminary technical results.  Recall that
in each of \csects{low}{high} the final configuration was determined by a
special family of sites; these families were denoted $P$ and $Q$
respectively, and were stationary during the evolution.  In the
intermediate region we need to define, somewhat similarly, two families
of sites, which we will denote by $P=\{p_k\}_{k\in\bbz}$ and
$Q=\{q_k\}_{k\in\bbz}$; here, however, the sites in these families move
during the evolution.

\begin{definition}\label{PQ}Suppose that we are given a height profile
$h(i)$, $i\in\bbz$, corresponding via \eqref{eta-h} to a configuration
$\eta$ of density $\rho$, $1/2<\rho<2/3$; as usual we let
$h^*(i)=h(i)+i/3$.  Let $A,B\,\bigl(=A(\eta),B(\eta)\bigr)\subset\bbz$ be
the sets of those sites which satisfy respectively
 \be\label{repeat}
 h(a)>\sup_{r>a}h(r),\ a\in A,\qquad\hbox{and}\qquad
  h^*(b)<\inf_{r>b}h^*(r),\ b\in B.
 \ee
 $A$ and $B$ are disjoint, since if $r\in A\cap B$ then \eqref{repeat}
implies that $h(r)>h(r+1)>h(r)-1/3$, impossible since $h$ takes integer
values.  Moreover, $A$ and $B$ are unbounded, both above and below, by
our assumption on $\rho$; we index the points of $A$ and $B$ as
increasing sequences $(a_j)_{j\in\bbz}$ and $(b_j)_{j\in\bbz}$,
respectively.  Let $P$ be the set of elements $a_j\in A$ such that there
exists a $b\in B$ satisfying $a_{j-1}<b<a_j$, and similarly let $Q$ be
the set of $b_j\in B$ such that there exists an $a\in A$ satisfying
$b_{j-1}<a<b_j$.  If we index $P$ as the increasing sequence
$(p_k)_{k\in\bbz}$, then clearly exactly one point of $Q$---the smallest
element of $B\cap(p_{k-1},p_k)$---lies in $(p_{k-1},p_k)$.  We denote
this element $q_{k-1}$.\end{definition} 

\begin{lemma}\label{sequences} The sequences $p=(p_k)_{k\in\bbz}$ and
$q=(q_k)_{k\in\bbz}$  satisfy
 \be\label{order}
 \cdots p_{-1}<q_{-1}<p_0<q_0<p_1<q_1\cdots
 \ee
 and
 \begin{align}
 h(p_k)&>\sup_{r>p_k}h(r),\label{p1}\\
 h(p_k)&\ge\sup_{q_{k-1}\le r\le p_k}h(r),\label{p2}\\
 h^*(q_k)&<\inf_{r>q_k}h^*(r),\label{q1}\\
 h^*(q_k)&\le\inf_{p_{k}\le r\le q_k}h^*(r).\label{q2}
 \end{align}
 Moreover they are, up to a shift of labels, the unique sequences
 satisfying these equations.
\end{lemma} 

A graphical interpretation of \eqref{p1}--\eqref{q2} is given in
\cfig{fig:midprof1}, drawn for future purposes at time 0.  Here $h=h_0$
is represented as the piecewise linear curve obtained by connecting each
pair of points $(k,h(k))$ and $(k+1,h(k+1))$ by a line segment.  We have
introduced also two families of straight lines: for each $k\in\bbz$,
$L_k$ is a horizontal line through $(p_{0,k},h_0(p_{0,k}))$, and $L_k^*$
a line of slope $-1/3$ through $(q_{0,k},h_0(q_{0,k}))$. \eqref{p1} and
\eqref{p2} imply respectively that the profile must lie below $L_k$
between $q_{k-1}$ and $p_k$ and strictly below $L_k$ to the right of
$p_k$.  Similarly, \eqref{q1} and \eqref{q2} imply that the profile lies
above $L_k^*$ between $p_k$ and $q_k$ and strictly above $L_k^*$ to the
right of $q_k$.

\begin{figure}[ht!]  \begin{center}
\includegraphics[width=5.3truein,height=2.8truein]{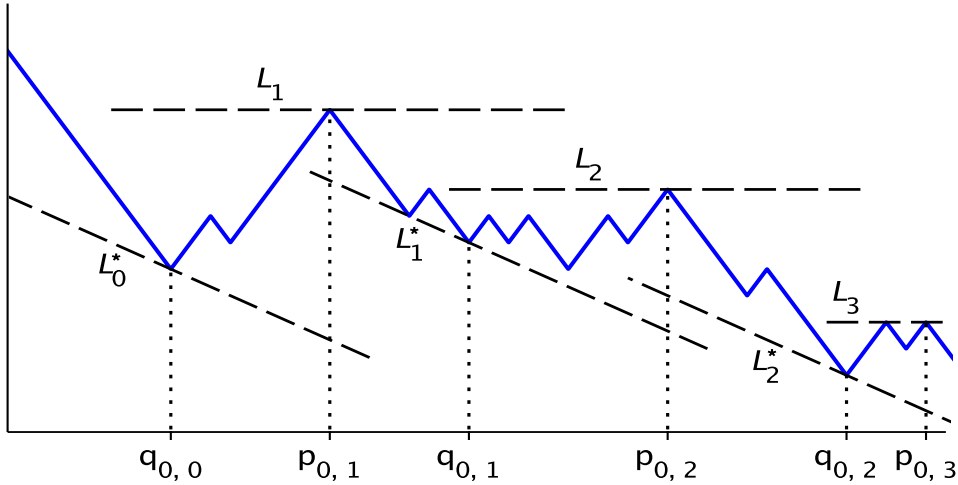}
\caption{Portion of typical initial height profile $h_0$ (solid blue
line), with the initial $p_k$ and $q_k$ values and the lines $L_k$
and $L_k^*$.}\label{fig:midprof1} \end{center} \end{figure}

\begin{proof}[Proof of \clem{sequences}]$(p_k)$ and $(q_k)$ clearly
satisfy \eqref{order}, \eqref{p1}, and \eqref{q1}.  \eqref{p2} follows
immediately from the fact that there can be no point of $A$ between
$q_{k-1}$ and $p_k$; the proof of \eqref{q2} is similar.

For uniqueness, suppose that $(\hat p_k)_{k\in\bbz}$ and
$(\hat q_k)_{k\in\bbz}$ satisfy \eqref{order}--\eqref{q2}, and let
$\Ph=\{\hat p_k\mid k\in\bbz$\} and $\Qh =\{\hat q_k\mid k\in\bbz$\}.
From \eqref{p1} and \eqref{q1} we see that $\Ph\subset A$ and
$\Qh\subset B$.  Now note that, for any $k$, \eqref{p2} implies that no
point of $A$ can belong to $(\hat q_{k-1},\hat p_k)$; this immediately
yields $\hat p_k\in P$.  Similarly, \eqref{q2} implies that
$\hat q_k\in Q$, so that $\Ph\subset P$ and $\Qh\subset Q$.  But now we
may again use $A\cap(\hat q_{k-1},\hat p_k)=\emptyset$ to conclude that
no point of $P$, and hence by \eqref{order} no point of $Q$, can lie in
$(\hat q_{k-1},\hat p_k)$; similarly, no point of $P$ or $Q$ can lie in
$(\hat p_k,\hat q_k)$, so that $P=\hat P$ and $Q=\hat Q$. \end{proof}

In our next result we record some trivial consequences of \clem{sequences}.

\begin{lemma}\label{brief}If $p$ and $q$ are as in \clem{sequences}
then for any $k\in\bbz$ we have (a)~$h^*(p_k+2)=h^*(p_k)-4/3$,
(b)~$h(q_k+1)=h(q_k)+1$, (c)~$h^*(q_k)\le h^*(p_k+2)$, and
(d)~$h(p_{k+1})\ge h(q_k+1)$. \end{lemma}

\begin{proof}Observe first that $\eta(p_k+1)=\eta(p_k+2)=1$, for
otherwise $h(p_k+2)\ge h(p_k)$, contradicting \eqref{p1}; this gives
(a).  Similarly, $\eta(q_k+1)=0$, since otherwise $h^*(q_k+1)<h^*(q_k)$,
contradicting \eqref{q1}, and this implies (b). (c) is an immediate
consequence of \eqref{q1} and \eqref{q2}, and (d) of \eqref{p2}.\end{proof}

We now turn to the dynamics.  We fix an initial configuration $\eta_0$,
with density $\rho$ satisfying $1/2<\rho<2/3$ (see \eqref{rhodef}); this
then evolves via $\eta_{t+1}=\Ui\eta_t$ and $h_{t+1}=\Ui h_t$.  Let
$A_t$, $B_t$, $P_t$, $Q_t$, $a_t$, $b_t$, $p_t$ and $q_t$ denote the sets
and sequences obtained from $h_t$ as in the definitions above.

\begin{remark}\label{lines} Before giving any further proofs we give a
brief qualitative description of the evolution of $h_t$; a key role is
played by the lines $L_k$, $L_k^*$ of \cfig{fig:midprof1}.  During the
evolution, the point $(p_{t,k},h(p_{t,k}))$ travels to the left along
$L_k$, moving either zero or two lattice sites at each time step and
stopping just short of the intersection with $L_k^*$.  Similarly,
$(q_{t,k},h_t(q_{t,k}))$ travels up and to the left along $L_k^*$, zero
or three lattice sites at each time step, and stops just short of
intersection with $L_{k+1}$ ($h^*_t(q_{t,k})$ is constant during this
evolution).  The precise limiting values $p_{\infty,k}$ and
$q_{\infty,k}$ are given in \eqref{pinf} and \eqref{qinf} below.  After
these special points have reached their limiting positions the profile
may continue to evolve between them, eventually reaching a limiting
position everywhere.  In the region between $p_{\infty,k}$ and
$q_{\infty,k}$, the limiting configuration has the form
$0\,1\,1\,0\,1\,1\,\cdots$ and $h_\infty$ has average slope $-1/3$, while
between $q_{\infty,k-1}$ and $p_{\infty,k}$ the form is
$1\,0\,1\,0\,\cdots$ and $h_\infty$ is essentially flat.  The limiting
configuration for the initial condition of \cfig{fig:midprof1} is shown
in \cfig{fig:midprof2}.  \end{remark}

\begin{figure}[ht!]  \begin{center}
\includegraphics[width=5.3truein,height=2.8truein]{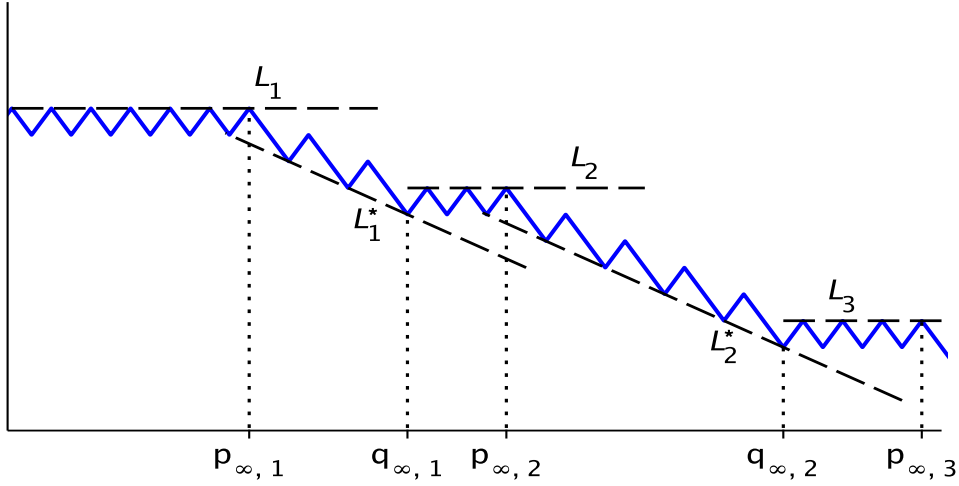}
\caption{Portion of final height profile for the initial profile of
\cfig{fig:midprof1}, with the final $p_k$ and $q_k$ values.}
\label{fig:midprof2} \end{center} \end{figure}

Next we show that (with appropriate indexing) the points
$(p_{t,k},h_t(p_{t,k}))$ and $(q_{t,k},h_t^*(q_{t,k}))$ move during the
evolution as described in \crem{lines}.

\begin{lemma}\label{evolve}The sequences $p_t$ and $q_t$ may be indexed
 so that for all $t\ge0$, 
 \begin{align}
(p_{t+1,k},h_{t+1}(p_{t+1,k}))&=(p_{t,k},h_t({p_t,k}))\text{ or }
   (p_{t,k}-2,h_t({p_t,k})),\label{pmove}\\
(q_{t+1,k},h_{t+1}^*(q_{t+1,k}))&=(q_{t,k},h_t^*({q_t,k})) \text{ or }
  (q_{t,k}-3,h_t^*({q_t,k})).\label{qmove1}
 \end{align}
 In particular, for each $k\in\bbz$ the sequences
$(p_{t,k})_{t=0}^\infty$ and $(q_{t,k})_{t=0}^\infty$ are nonincreasing
and the sequences $(h_t(p_{t,k}))_{t=0}^\infty$ and
$(h_t^*(q_{t,k}))_{t=0}^\infty$ constant.  \end{lemma}

\begin{proof}Examination of the action of the dynamics near the $p_k$ and
$q_k$ suggests that appropriate indexing will yield $p_{t+1,k}=p'_k$ and
$q_{t+1,k}=q'_k$, where
 \begin{align}\label{pkprime}
  p'_k&=\begin{cases}p_{t,k},& \text{if 
   $\eta_t\rng{p_{t,k}+1}{p_{t,k}+3}=1\,1\,0$,}\\
 p_{t,k}-2,& \text{if
   $\eta_t\rng{p_{t,k}+1}{p_{t,k}+3}=1\,1\,1$,}\end{cases}\\
 \label{qkprime}
 q'_k&=\begin{cases}q_{t,k},& \text{if
   $\eta_t\rng{q_{t,k}+1}{q_{t,k}+3}=0\,1$,}\\
 q_{t,k}-3,& \text{if
   $\eta_t\rng{q_{t,k}+1}{q_{t,k}+3}=0\,0$,}\end{cases}
 \end{align}
 (in each case the given possibilities are exhaustive).  To verify this,
and hence prove the result (for one sees easily that
$h_{t+1}(p_k')=h_t(p_{t,k})$ and $h_{t+1}^*(q_k')=h_t^*(q_{t,k})$) it
suffices, by the uniqueness in \clem{sequences}, to check that $(p_k')$
and $(q_k')$ satisfy \eqref{order}--\eqref{q2}.

Now \eqref{order}--\eqref{q2} imply that if $p'_k=p_{t,k}-2$ then either
$q_{k-1}'=q_{t,k-1}-3$ or $q_{t,k-1}\le p_{t,k}-3$, and that if
$q_k'=q_{t,k}-3$ then either $p_{t,k}\le q_{t,k}-4$ or
$p_{t,k}=q_{t,k}-3$ but $p'_k=p_{t,k}-2$; \eqref{order} for $(p_k')$ and
$(q_k')$ follows.  To continue, recall that the $\Ui$ dynamics takes
place in two steps, with the usual F-TASEP dynamics, at which let us say
$h_t(i)$ becomes $H(i)$, followed by a two-site translation to the left
(we also write $H^*(i)=H(i)+i/3$).  Now for $i>p_{t,k}$,
$H(i)\le h_t(p_{t,k})$, with equality only if $i=p_{t,k}+2$ and
$\eta_t\rng{p_{t,k}+1}{p_{t,k}+3}=1\,1\,0$, and it is precisely in this
case that $i$ becomes $p'_k$ after the translation.  Thus \eqref{p1} is
satisfied for $p_k'$.  \eqref{q1} for $q_k'$ is checked similarly.

One can check \eqref{p2} and \eqref{q2} considering separately the
various cases of \eqref{pkprime} and \eqref{qkprime}.  To illustrate,
consider \eqref{q2} when $p_k'=p_{t,k}-2$ and $q_k'=q_{t,k}$.  If
$p_{t,k}\le i\le q_{t,k}$ and $h^*_t(i)=h^*(q_{t,k})$ then from
\eqref{q2} at time $t$ necessarily $\eta\rng{i-1}{i+1}=1\,1\,0$, so that
$H^*(i)=h^*_t(i)+2$, and this with \eqref{q1} implies that
$H^*(i)\ge h_t^*(q_{t,k})+2/3$ for $p_{t,k}\le i\le q_{t,k}+2$. After the
translation step this becomes $H^*(i)\ge h_t^*(q_{t,k})$ for
$p_k'\le i\le q_k'$, verifying \eqref{q2} in this case.  \end{proof}

\begin{theorem}\label{midlimit}For each $i,k\in\bbz$, $\eta_t(i)$,
$h_t(i)$, $p_{t,k}$, and $q_{t,k}$ are eventually constant.  If we denote
these limiting values by $\eta_\infty(i)$, $h_\infty(i)$, $p_{\infty,k}$,
and $q_{\infty,k}$, then $p_\infty$ and $q_\infty$ are the sequences
obtained from $h_\infty$ as in \cdef{PQ}, and $\eta_\infty$ is given by
 \begin{align}
 \eta_\infty\rng{q_{\infty,k-1}+1}{p_{\infty,k}}=0\,1\,0\,1\,\ldots0\,1\,0,
   \label{01}\\
 \eta_\infty\rng{p_{\infty,k}+1}{q_{\infty,k}}=1\,1\,0\,1\,1\,\ldots0\,1\,1.
 \label{011}
 \end{align}
 Moreover,
 \begin{align}
p_{\infty,k}&=q_{0,k}-3(h_0(p_{0,k})-h_0(q_{0,k}))+4,\label{pinf}\\
q_{\infty,k}&=p_{0,k+1}-3(h_0^*(p_{0,k+1})-h_0^*(q_{0,k}))+3\label{qinf}.
\end{align}
   \end{theorem}

We can summarize the theorem thus: the final configuration $\eta_\infty$
has the form 
 \be\label{midetainf}
\eta_\infty=\cdots(0\,1)^{n_k}(0\,1\,1)^{m_k}(0\,1)^{n_{k+1}}
    (0\,1\,1)^{m_{k+1}}\cdots,
 \ee
 with 
 \be\label{nm}\begin{aligned}
 2n_k&=p_{\infty,k}-q_{\infty,k-1}-1=3(h_0^*(q_{0,k})-h_0^*(q_{0,k-1})),\\
 3m_k&=q_{\infty,k}-p_{\infty,k}+1=3(h_0(p_{0,k+1})-h_0(p_{0,k})).
 \end{aligned}\ee
 These results are illustrated in \cfig{fig:midprof3}.
 
\begin{figure}[ht!]  \begin{center}
\includegraphics[width=5.3truein,height=2.00truein]{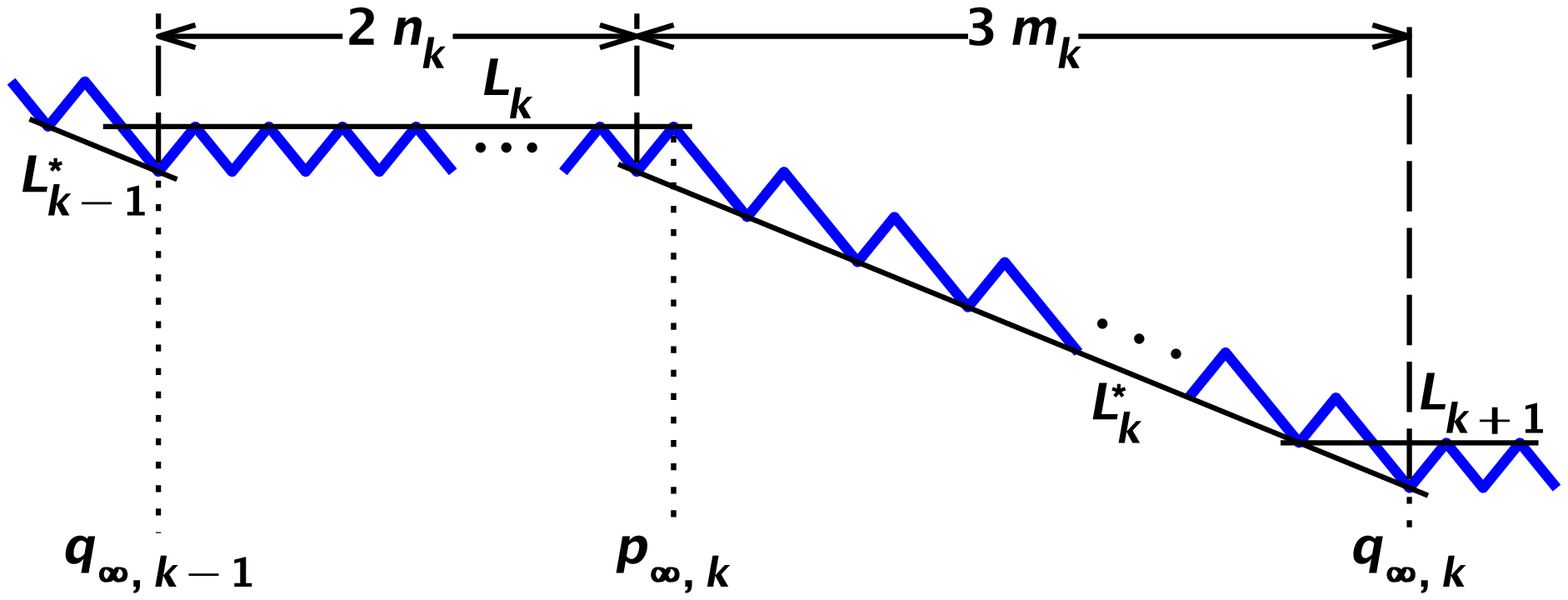}
\caption{Illustration of equations \eqref{01}, \eqref{011}, and
 \eqref{midetainf}.}
\label{fig:midprof3} \end{center} \end{figure}
 
\begin{proof}[Proof of \cthm{midlimit}]The nonincreasing sequences
$(p_{t,k})_{t=0}^\infty$ and $(q_{t,k})_{t=0}^\infty$ are clearly bounded
below, since, for example, $\bigl(q_{t,k},h_t(q_{t,k})\bigr)$ must remain
on the line $L_k^*$, and by \eqref{p2} must stay below $L_k$.  Thus the
limits $p_{\infty,k}$ and $q_{\infty,k}$ exist and will be attained by
some finite time.  Suppose that $t$ is a time for which $q_{t,k-1}$,
$p_{t,k}$, and $q_{t,k}$ have reached their limiting values.

To establish \eqref{01}, note first that if $\eta_t$ satisfies
\eqref{01} then this will remain true as $t$ increases.  Moreover, if
\eqref{01} does not hold (for $\eta_t$) then by \eqref{p2} we just have
that
 \be\label{SG}
\eta_t\rng{q_{\infty,k-1}+1}{p_{\infty,k}+3}
   = \cdots\,0\,0\,(10)^j\,1\,1\,0
 \ee
 for some $j\ge0$.  But then $\eta_{t+1}$ must either satisfy \eqref{01}
or be of the form \eqref{SG} with $j$ replaced by some $j'\ge j+1$.
Thus \eqref{01} must be attained in finite time.  \eqref{011} is obtained
similarly, with \eqref{SG} replaced by
 \be
 \eta_t\rng{p_{\infty,k}+1}{q_{\infty,k}+2}
  = \cdots\,1\,1\,1\,(011)^j\,0\,1.
 \ee

Finally, it follows from \eqref{01} that $q_{\infty,k}=i+3$, where $i$ is
the site at which the lines $L^*_{k-1}$ and $L_k$ intersect, and this is
just \eqref{qinf}.  Similarly, \eqref{011} implies that
$p_{\infty,k}=i'+4$, with $i'$ is the intersection of $L_k$ and $L^*_k$,
yielding \eqref{pinf}.
\end{proof}
 
\subsection{An  initial Bernoulli distribution\label{interbern}}

We again take up the case in which the initial configuration $\eta_0$ is
distributed according to the Bernoulli measure $\murho$, now with
$1/2<\rho<2/3$, and ask for the distribution of the final configuration
$\eta_\infty$, which we will obtain from the joint distribution of the
random variables $n_k$ and $m_k$ of \eqref{midetainf} (once these are
precisely defined---compare \cthm{ldrenewal}).  Note that these variables
are expressed in \eqref{nm} as functions of the initial configuration; we
will hence in this section refer to properties of the initial
configuration only, and write simply $\eta$, $h$, $p_k$, and $q_k$ rather
than $\eta_0$, etc.  While the process $\ldots,n_k,m_k,n_{k+1}\ldots$ is
not Markovian, we will show that one may define a ``hidden'' Markov
process, determined by the initial configuration, such that the variables
$n_k$ and $m_k$ are functions of the variables of that process.

To obtain a well-defined labeling of the points of $P$ and $Q$ we
introduce $V:=\{\eta\in X_\rho\mid0\in P(\eta)\}$, defining $p_0(\eta)=0$
for $\eta\in V$ and labeling the remaining points of $P(\eta)$ and
$Q(\eta)$ to satisfy \eqref{order}.  We write $\mu$ for the measure
$\murho$ conditioned on $V$.  We also decompose $X_\rho$ as
$X_\rho=X_\rho^-\times X_\rho^+$, where $X_\rho^-\subset\{0,1\}^{\bbz_-}$
is the set of configurations $\alpha:\bbz_-\to\{0,1\}$ which satisfy
$\lim_{N\to\infty}(N+1)^{-1}\sum_{i=-N}^{0}\alpha(i) = \rho$, and
$X_\rho^+\subset\{0,1\}^{\bbn}$ is the set of configurations
$\beta:\bbn\to\{0,1\}$ which satisfy
$\lim_{N\to\infty}N^{-1}\sum_{i=1}^N\beta(i)=\rho$.  We correspondingly
write $\eta\in X_\rho$ as $(\eta^-,\eta^+)$.

Now suppose that $F\subset V$ is an event, with $\mu(F)>0$, specifying an
arbitrary amount of information about $\eta^-$ and the $p_k$, $h(p_k)$,
$q_k$, and $h(q_k)$ for $k<0$, including in particular the values of
$q_{-1}=\tilde q$ and $h^*(q_{-1})=\tilde h^*$, while $F_0\subset V$
specifies only $q_{-1}=\tilde q$ and $h^*(q_{-1})=\tilde h^*$ (note that
the nonlocality in \cdef{PQ} means that $\eta^-$ does not determine the
$p_k$ and $q_k$, $k<0$).  Clearly from \eqref{p1} and \eqref{q1}, and the
fact that $p_0=h(p_0)=0$ on $V$, the occurrence of either $F$ or $F_0$
implies that $I$ occurs, where
 \be\label{restrict}I\,(=I_{\tilde h^*})
  :=\{\eta\mid h(i)<0 \text{ and } h^*(i)>\tilde h^*
    \text{ for all $i\ge1$}\}.
 \ee
 The next result gives the basic Markovian property of the $p_k$'s
and $q_k$'s.

\begin{lemma}\label{Markovian} The distribution of $\eta^+$ when
conditioned on $F$ is the same as when conditioned on $F_0$.  Moreover,
this distribution is explicitly given by the marginal of $\murho$ on
$X_\rho^+$,  conditioned on $I$.\end{lemma}

In preparation for the proof we make a preliminary definition: for
$\alpha\in X_\rho^-$ we adapt \cdef{PQ} to define
$A'(\alpha):=\{a\le0\mid h(a)>\sup_{a<r\le0}h(r)\}$,
$B'(\alpha):=\{b<0\mid h^*(b)<\inf_{b<r\le0}h^*(r)\}$ (so that
$0\in A'(\alpha)$, $0\notin B'(\alpha)$) and obtain $P'(\alpha)$ and
$Q'(\alpha)$ from $A'(\alpha)$ and $B'(\alpha)$ in parallel with
\cdef{PQ}; we index the elements of these sets as $(p'_k)_{k\le0}$ and
$(q'_k)_{k<0}$, with $p_0'=\max P'(\alpha)$.  We view $P'(\eta^-)$ and
$Q'(\eta^-)$ as approximations to $P(\eta)\cap\bbz_-$ and
$Q(\eta)\cap\bbz-$ which depend only on $\eta^-$.  But we have

\begin{lemma}\label{PQprime} If $\eta\in V$ then
$P'(\eta^-)=P(\eta)\cap\bbz_-$ and
$Q'(\eta^-)=Q(\eta)\cap\bbz_-$.\end{lemma}

\begin{proof} Clearly $B(\eta)\subset B'(\eta^-)$ for all $\eta\in X$ and
if $\eta\in V$ then $A'(\eta^-)=A(\eta)$.  We index the points of
$A(\eta)$ and $A'(\eta^-)$ so that $a_0=a'_0=0$.  Then
$q_{-1}(\eta)=\min\bigl(B(\eta)\cap(a_{-1},a_0)\bigr)$ and since, by
\eqref{q2}, $h^*(r)\ge h^*(q_{-1})$ for $a_{-1}\le r<q_{-1}$,
$B'(\eta^-)\cap(a_{-1},q_{-1})=\emptyset$.  From this we find easily that
$B'(\eta^-)\cap(-\infty,q_{-1}]=B(\eta)\cap(-\infty,q_{-1}]$ and the
result follows.\end{proof}

\begin{proof}[Proof of \clem{Markovian}]Without loss of generality we may
assume that $F$ has the form $V\cap G_L\cap H_N$, where $G_L$ specifies
$\eta(i)$ for $-L\le i\le0$ and $H_N$ specifies $(p_k,h(p_k))$ for
$-N\le k\le0$ and $(q_k,h^*(q_k))$ for $-N\le k<0$, and in particular
requires that $p_0=h(p_0)=0$, $q_{-1}=\tilde q$, and
$h(q_{-1})=\tilde h^*$.  We claim that $F=G_L\cap H_N'\cap I$, where
$H'_N$ gives the same specification to the $p'_k(\eta^-)$ and
$q'_k(\eta^-)$ that $H_N$ gave to the $p_k$ and $q_k$.  Assuming this,
for $J$ an arbitrary event depending only on $\eta(i)$ for $i\ge1$ we have,
using first $F\subset V$ and then
$\murho(F)=\murho(G_L\cap H'_N)\murho(I)$ and
$\murho(J\cap F)=\murho(G_L\cap H'_N)\murho(J\cap I)$,
 \be
\mu(J\mid F)=\murho(J\mid F)=\murho(J\mid I),
 \ee
 which is the desired conclusion.

To verify that $F=G_L\cap H_N'\cap I$ we observe first that if
$\eta\in F$ then \eqref{p1} and \eqref{q1} imply that $\eta\in I$;
moreover, by \clem{PQprime}, $\eta\in H_N'$; thus
$F\subset G_L\cap H_N'\cap I$.  Conversely, if
$\eta\in G_L\cap H_N'\cap I$ then $\eta\in I$ implies that $0\in A(\eta)$
and, with $\tilde q\in B'(\eta^-)$, that $\tilde q\in B(\eta)$.
Moreover, from $\eta\in H_N'$ it follows that $0\in P'(\eta^-)$, and with
$\eta\in I$ and $\tilde q\in Q'(\eta^-)$ this implies that
$\eta\in V$; from this, $G_L\cap H_N'\cap I\subset F$ is immediate.
\end{proof}

A similar result holds with the roles of the $p_k$ and $q_k$
interchanged. Let $V^*:=\{\eta\in X_\rho\mid0\in Q(\eta)\}$, index the
points of $P$ and $Q$ on $V^*$ via $q_0(\eta)=0$ and \eqref{order}, and
let $\mu^*$ be $\murho$ conditioned on $V^*$.  Suppose that
$F^*\subset V^*$ is an event, with $\mu(F^*)>0$, specifying an arbitrary
amount of information about $\eta^-$ and the $p_k$, $h(p_k)$, $q_k$, and
$h(q_k)$ for $k<0$, including in particular the values of $p_{-1}=\hat p$
and $h(p_{-1})=\hat h$, while $F_0^*\subset V$ specifies only
$p_{-1}=\hat p$ and $h(p_{-1})=\hat h$.  The occurrence of either $F^*$
or $F_0^*$ implies that $I^*$ occurs, where
 \be\label{restrict2} I^*\,(=I^*_{\hat h})
    :=\{\eta\mid h^*(i)>0 \text{ and } h(i)<\hat h
    \text{ for all $i\ge1$}\}.
 \ee
 The proof of the next result is parallel to that of \clem{Markovian}.

\begin{lemma}\label{Markovian2} The distribution of $\eta^+$ when
conditioned on $F^*$ is the same as when conditioned on $F_0^*$.
Moreover, this distribution is explicitly given by the marginal of
$\murho$ on $X_\rho^+$, conditioned on $I^*$.\end{lemma}

We next turn to the definition of the Markov process.  Let
$(Y_j)_{j\in\bbz}$ be the sequence of random variables on $V$ which take
values in $\bbz^2$ and are defined for $k\in\bbz$ by
 \be\label{markov}\begin{aligned}
  Y_{2k}&=\bigl(p_k-q_{k-1},h(p_k)-h(q_{k-1})\bigr),\\
Y_{2k+1}&=\bigl(q_k-p_k,h(q_k)-h(p_k)\bigr),\\
 \end{aligned}\ee
 This definition seems to single out $p_0$ (among the points of
$P\cup Q$) to play a special role, but the next lemma shows that this is
not really the case.

 \begin{lemma}\label{transn} (a) Fix $k\in\bbz$ and define the variables
$Y'_j$, $j\in\bbz$, on $V$ by $Y_j'=Y_{j+2k}$.  Then the joint
distribution of $(Y'_j)_{j\in\bbz}$ is the same as that of
$(Y_j)_{j\in\bbz}$.
 \par\smallskip\noindent
 (b) Suppose that $(Y^*_j)_{j\in\bbz}$ is defined on $V^*$ by replacing
$Y$ by $Y^*$ in \eqref{markov}.  Then $(Y^*_j)_{j\in\bbz}$ and
$(Y_j)_{j\in\bbz}$ have the same joint distribution.  \end{lemma}

\begin{proof}For (a) it suffices to show that the distribution of 
$\tau^{-p_k}\eta$, the configuration seen from $p_k$, is the same as $\mu$
itself.  But this measure is
 \begin{align*}
\murho(V)^{-1}\sum_{i\in\bbz}\tau_*^{-i}\murho\big|_{V\cap\{p_k=i\}}
 &=\murho(V)^{-1}\sum_{i\in\bbz}\murho\big|_{\tau^{-i}(V\cap\{p_k=i\})}\\
 &=\murho(V)^{-1}\murho\big|_{\bigcup_{i\in\bbz}\tau^{-i}(V\cap\{p_k=i\})}\\
 &=\murho(V)^{-1}\murho\big|_V
  =\mu.  \end{align*}
 Replacing $p_k$ by $q_k$ in the above, and in the last line $V$ by $V^*$
 and $\mu$ by $\mu^*=\murho(V^*)^{-1}\mu^{(\rho)}\big|_{V^*}$, we obtain (b).
 \end{proof}

 \begin{theorem}\label{hidden} $(Y_j)_{j\in\bbz}$ is a Markov process.
\end{theorem}

\begin{proof}[Proof of \cthm{hidden}] We discuss first the transition
from $Y_0$ to $Y_1$. Observe that if $\eta\in V$ then $Y_1(\eta)$ is
determined by $Y_0(\eta)=\bigl(-q_{-1}(\eta),-h(q_{-1})\bigr)$ and
$\eta^+$, for certainly $B(\eta)\cap(0,\infty)$ is determined by $\eta^+$
and then since $p_0(\eta)=0$,
$q_0(\eta)=\min\left(B(\eta)\cap(0,\infty)\right)$. But by
\clem{Markovian} no knowledge of $Y_j$, $j<0$, can affect the
distribution of $\eta^+$ determined by $Y_0$; this is the Markov
property.  \clem{transn}(a) then implies that transitions from $Y_{2k}$
to $Y_{2k+1}$, $k\in\bbz$, are all Markovian.  That the transitions from
$Y_{2k-1}$ to $Y_{2k}$ are also Markovian follows from \clem{transn}(b)
and an argument on $V^*$ similar to the above. \end{proof}

 There are two transition matrices for this Markov process, for odd and
even steps respectively.  These can be expressed in terms of
combinatorial quantities $e_{a,b}^{n,m}$ which generalize the Catalan and
Fuss-Catalan numbers encountered earlier (although we don't have
closed-form expressions for these quantities).  Here $a$, $m$, and $n$
are integers, with $a\ge-1$ and $n\ge0$, and $b$ is of the form $2l/3$
with $l$ an integer (see (\ref{highmod}.ii)) and $b\le4/3$.
$e_{a,b}^{n,m}$ counts the number of (partial) height profiles
$h:\{0,\ldots,n\}\to\bbz$, with $h(i+1)-h(i)=\pm1$ for $i=0,\ldots,n-1$,
which satisfy
 \be\label{Edef}
h(0)=0,\ h(n)=m, \text{ and }
     b-\frac{i}3 \le h(i) \le a,\; 1\le i\le n.
 \ee
Note that if $b\le-2n/3$, so that the left-hand inequality in \eqref{Edef}
is satisfied for all possible $h$, then $e_{0,b}^{2n,0}=c_n$, and
similarly that for $a\ge n$, $e_{a,0}^{3n,-n}=d_n$ (see \eqref{catalan}
and \eqref{raney}).  Thus $e_{a,b}$ is a generalization of the Catalan
and Fuss-Catalan sequences which allows for appropriate upper and lower
bounds on the profiles.  Note that if we consider these profiles as
arising from configurations in $\{0,1\}^{\{1,\ldots,n\}}$ and weight these
configurations with a Bernoulli product measure of density $\rho$ then
the set of configurations counted by $e_{a,b}^{n,m}$ has probability
$f_{a,b}^{n,m}:=e_{a,b}^{n,m}\rho^{(n-m)/2}(1-\rho)^{(n+m)/2}$.

We calculate the transition matrix $\M^0$ from $Y_0$ to $Y_1$ (which is
the matrix for any transition $Y_{2n}\to Y_{2n+1}$) by taking
$Y_0=(-\tilde q,-\tilde h^*+\tilde q/3)$ and using the marginal on
$X_\rho^+$ of the conditional measure $\murho(\cdot\mid I_{\tilde h^*})$;
to obtain the matrix $\M^1$ for the transition from $Y_1$ to $Y_2$ (or
$Y_{2n+1}\to Y_{2n+2}$) we take $Y_1=(-\hat p,-\hat h)$ and use the
marginal on $X_\rho^+$ of $\murho(\cdot\mid I_{\hat h}^*)$.  To compute
the normalization $\murho(I_{\tilde h^*})$ we note that a partial profile
$h(i)_{i=1}^n$ obeys the bounds defining $I_{\tilde h^*}$ and passes
through $(n,m)$ iff it satisfies \eqref{Edef} with $a=-1$ and
$b=\tilde h^*+2/3$.  Thus there are $e_{-1,\tilde h^*+2/3}^{n,m}$ such
profiles; each has probability $\rho^{(n-m)/2}(1-\rho)^{(n+m)/2}$ so that
 \be
 \murho(I_{\tilde h^*})=\lim_{n\to\infty}\sum_{m=\tilde h^*+2/3-n/3}^{-1}
   f_{-1,\tilde h^*+2/3}^{n,m}.
 \ee
 To obtain $\murho(I^*_{\hat h})$, note that the restrictions
corresponding to the bounds defining $I_{\hat h}^*$ are given by
\eqref{Edef} with $a=\hat h-1$ and $b=2/3$, so that
 \be\label{I*}
\murho(I^*_{\hat h})=\lim_{n\to\infty}\sum_{m=2/3-n/3}^{\hat h-1}
   f_{\hat h-1,2/3}^{n,m}.
 \ee

We can now write down the transition matrix $\M^0_{y_0,y_1}$ for the
transition $Y_0\to Y_1$ (and any $Y_{2n}\to Y_{2n+1}$).  Set
$y_0=(-\tilde q,-\tilde h^*+\tilde q/3)$ as above and $y_1=(q',h')$, and
note that $\M^0_{y_0,y_1}$ vanishes unless
$\tilde h^*+2/3-q'/3\le h'\le-2$.  When this condition is satisfied, a
configuration $\eta^+$ with height function $h$ contributes to
$\mu(Y_1=y_1\mid Y_0=y_0)=\murho(Y_1=y_1\mid I_{\tilde h^*})$ iff:
(i)~$h$ reaches $(q',h')$ while obeying the restrictions specified by
\eqref{Edef} with the replacements $n\to q'$, $m\to h'$, $a\to-1$, and
$b\to h'+q'/3$, and (ii)~$h$ satisfies $h(i)\le -1$ and
$h^*(i)\ge h'+q'/3+2/3$ for $i\ge q'+1$, that is, the tail of $h$ is a
translate of a profile contributing to $I^*_{-h'}$ (see \eqref{I*} and
preceding discussion).  Thus if $\tilde h^*+2/3-q'/3\le h'\le-2$,
 \be\label{M0}
 \M^0_{y_0,y_1}=\frac{f_{-1,h'+q'/3}^{q',h'}\,
  \murho(I^*_{-h'})}{\murho(I_{\tilde h^*})}
 =\frac{f_{-1,y_{1,2}+y_{1,1}/\sss3}^{y_{1,1},y_{1,2}}\,
  \murho(I^*_{-y_{1,2}})}{\murho(I_{-y_{0,2}-y_{0,1}/\sss3})}.
 \ee
 A similar calculation gives the matrix for transitions $Y_1\to Y_2$ (and
any $Y_{2n+1}\to Y_{2n+2}$); taking $y_1=(-\hat p,-\hat h)$ and
$y_2=(p'',h'')$ we see that $\M^1_{y_1,y_2}$ vanishes unless
$2-p''/3\le h''\le\hat h-1$, and when this is satisfied,
 \be\label{M1}
 \M^1_{y_1,y_2}=\frac{f_{h'',2/3}^{p'',h''}\,
  \murho(I_{-h''-p''/3})}{\murho(I^*_{\hat h})}
 =\frac{f_{y_{2,2},\sss2/3}^{y_{2,1},y_{2,2}}\,
  \murho(I_{-y_{2,2}-y_{2,1}/\sss3})}{\murho(I^*_{-y_{1,2}})}.
 \ee

Although we have not provided a very explicit expression for the
transition probability for the Markov chain, we can more explicitly
characterize this process as a Gibbs state.  Consider for example the
probability that $Y_n=y_n$ for $-2N+1\le n\le 2N$, given that
$Y_{-2N}=y_{-2N}$.  It follows from \eqref{M0} and \eqref{M1} (and even
more directly from the successive bounds on the height function $h$
implied by the history of the Markov chain) that this probability is
given (somewhat formally) by
 \be\label{bigprob}\begin{aligned}
\M^0_{y_{-2N},y_{-2N+1}}&\M^1_{y_{-2N+1},y_{-2N+2}}
\cdots\M^0_{y_{2N-2},y_{2N-1}}\M^1_{y_{2N-1},y_{2N}}\\
&\hskip-40pt= Z^{-1}\exp\left(-\sum_{n=-N}^{N-1}\bigl(v^0(y_{2n},y_{2n+1})
    +v^1(y_{2n+1},y_{2n+2})\bigr)\right),
 \end{aligned}\ee
 where
$Z=\murho(I_{-y_{-2N,2}-y_{-2N,1}/\sss3})/\murho(I_{-y_{2N,2}-y_{2N,1}/\sss3})$
and
 \begin{align*}
v^0(y_{2n},y_{2n+1})&=\begin{cases}
      -\log f_{-1,y_{2n+1,2}+y_{2n+1,1}/\sss3}^{y_{2n+1,1},y_{2n+1,2}},&\\
   \noalign{\vskip3pt}&\hskip-95pt
   \text{if $-y_{2n,2}+\ds\frac{2-y_{2n,1}-y_{2n+1,1}}3\le y_{2n+1,2}\le-2$,}\\
    \infty,&\hskip-95pt\text{otherwise;}\end{cases}\\
v^1(y_{2n+1,2n+2})&=\begin{cases}
    -\log f_{y_{2n+2,2},\sss2/3}^{y_{2n+2,1},y_{2n+2,2}},&\\
   \noalign{\vskip3pt}&\hskip-69pt
    \text{if $2-\ds\frac{y_{2n+2,1}}3\le y_{2n+2,2}\le-y_{2n+1,2}-1$},\\
   \infty,&\hskip-69pt\text{otherwise.}\end{cases}
 \end{align*}
 The two-sided conditional probability that $Y_n=y_n$ for
$-2N+1\le n\le 2N-1$, given that $Y_{-2N}=y_{-2N}$ and $Y_{2N}=y_{2N}$,
 is then given by the same formula \eqref{bigprob}, with $Z$ now a
 normalizing constant.  We can argue similarly for all two-sided
 conditional probabilities, and we thus see that our Markov chain is a
 Gibbs state with interaction potentials given by $v^0$ and $v^1$. 

\begin{remark}\label{six}Numerical simulations of the model in the
intermediate density region show convincingly that the two-point function
$g(k)$ in the final state satisfies an analogue of Theorems \ref{twopt}(b) and
\ref{hdvarious}(d): for any $n\ge0$, $\sum_{i=1}^6g(6n+i)=6\rho^2$.
We conjecture that this is in fact true, but have no proof at the moment.
\end{remark}

 \medskip\noindent
 {\bf Acknowledgments:} We thank Ivan Corwin and Pablo Ferrari for
helpful comments.  The work of JLL was supported by the AFOSR under
award number FA9500-16-1-0037.

\appendix

\section{Generating functions\label{genfct}}

Our goal is to calculate the generating function
$G(z):=\sum_{n=1}^\infty g(n)z^n$ of the two-point function in
the low density region; the generating function $G^T(z)$ of the
truncated two-point function (see \clem{decay} and its proof) is then
given by $G^T(z)=G(z)-\rho^2z/(1-z)$.  We will use the quantities
 \begin{align}\label{psi}
\psi(n)&=c_n\rho^n(1-\rho)^{n+1},
 &\Psi(u)&=\sum_{n\ge0}\psi(n)u^n=\frac{2(1-\rho)}{1+\sqrt{1-4\rho(1-\rho)u}},\\
 \theta(n)&=\sum_{m\ge n}\psi(m),\label{theta}
 &\Theta(u)&=\sum_{n\ge0}\theta(n)u^n=\frac{1-u\Psi(u)}{1-u}\\
 \lambda(n)&=\sum_{m\ge n}\theta(m),\label{lambda}
 &\Lambda(u)&=\sum_{n\ge0}\lambda(n)u^n=\frac{\Theta(1)-u\Theta(u)}{1-u}.
 \end{align}
 Here \eqref{psi} is obtained from a standard formula for Catalan series,
see e.g.  \cite{Catalan}.  In obtaining \eqref{theta} we have used
$\Psi(1)=1$, which follows from \eqref{psi} or from the normalization of
the distribution \eqref{catdist}.  From \eqref{psi}--\eqref{lambda} we
further obtain
 \be\label{conseq}
\Theta(1)=\lambda(0)=\frac{1-\rho}{1-2\rho}, \qquad 
  \lambda(1)=\Theta(1)-1=\frac\rho{1-2\rho}.
 \ee

Now write $g(n)=\sum_{m\ge0}g_m(n)$, where $g_m(n)$ is the contribution
to $g(n)$ from configurations in which $m$ points of $P$, say
$p_1<p_2<\cdots<p_m$, lie between sites $0$ and $n$; note that $g_m(n)=0$
unless $m$ and $n$ have the same parity.  We let $p_0=-(2n_0+1)$ be the
largest point of $P$ to the left of 0, and $p_{m+1}$ be the smallest
point of $P$ to the right of $n$.  We first consider the special case
$m=0$; with $n=2l$ and $p_1=2n_1$, $n_1>l$, we have
 \be
g_0(n)=(1-2\rho)\sum_{n_0\ge0}\sum_{n_1\ge l+1}\psi(n_0+n_1)
 = (1-2\rho)\lambda(l+1),
 \ee
  and then, using \eqref{conseq}
 \be\label{Gmez}
\sum_{l\ge1}g_0(2l)z^{2l}
   =\frac{1-2\rho}{z^2}
   \left(\Lambda(z^2)-\frac{z^2\rho}{1-2\rho}-\frac{1-\rho}{1-2\rho}\right).
 \ee

Now we turn to the case $m\ge1$, writing $p_1=2n_1$ with $n_1\ge1$,
$p_j-p_{j-1}=2n_j+1$ for $j=2,\ldots,m$, $n-p_m=2l+1$, and
$p_{m+1}-p_m=2n_{m+1}+1$ with $n_{m+1}>l$.  The contribution to $g_m(n)$
for fixed $p_1,\ldots,p_m$ is
 \begin{align}\nonumber
(1-2\rho)&\left(\sum_{n_0\ge0}\psi(n_0+n_1)\right)
 \prod_{j=2}^m\psi(n_j)\left(\sum_{n_{m+1}\ge l+1}\psi(n_{m+1})\right)\\
 &=(1-2\rho)\theta(n_1)\prod_{j=2}^m\psi(n_j)\theta(l+1).\label{step1}
 \end{align}
 Multiplying \eqref{step1} by $z^n$ and summing over $n$ and
$n_1,\ldots,n_m$, and then over $m$, yields
 \be\label{Gmgz}
 \sum_{m\ge1}\sum_{n\ge1} g_m(n)z^n
  = (1-2\rho)\frac{\bigl(\Theta(z^2)-1\bigr)^2}
 {z(1-z\Psi(z^2))}.
 \ee
 The generating function $G(z)$ is the sum of \eqref{Gmez} and
\eqref{Gmgz}.  

From the formulas above it is clear that the possible singularities of
$G(z)$ are at $z=\alpha_0:=(4\rho(1-\rho))^{-1/2}$, where
$\Psi(z^2)$ is singular, and at the unique root $z=1$ of
$z\Psi(z^2)=1$; this uniqueness may be verified, for example, from
the fact \cite{Catalan} that $\Psi$ satisfies
$\Psi(u)=1-\rho+u\rho\Psi^2(u)$.  (There is also a singularity at $z=0$,
but $\Psi(u)$ as defined in \eqref{psi} is clearly regular at $u=0$; this
singularity lies on the second sheet.)  A straightforward calculation
shows that $G(z)$ has a simple pole at $z=1$, with residue
$\rho^2$, and this pole is removed in passing to $G^T(z)$ via
 \be\label{GT}
 G^T(z)=G(z) - \frac{\rho^2z}{1-z}.
 \ee
  Thus $G^T(z)$ is
analytic for $|z|<\alpha_0$ (see \cthm{decay}).

\begin{remark}\label{secondproof}If one writes
$G(z)=G_{\rm even}(z)+G_{\rm odd}(z)$, where $G_{\rm even}$
and $G_{\rm odd}$ are respectively even and odd in $z$, then one
finds that 
$G_{\rm even}(z)+z G_{\rm odd}(z)=2\rho^2z^2/(1-z^2)$.
This is an independent proof of \cthm{twopt}(b). \end{remark}

\section{A semi-infinite system\label{semiinf}}

Consider again the system at low density.  In \csect{lowbern} we
introduced the event $F:=\{\eta_0\mid0\in P(\eta_0)\}$, where $P(\eta_0)$
was defined in \eqref{lowpk}; $F$ is invariant under the F-TASEP dynamics
and, under that dynamics on $F$, no particles jump from site 0 to site 1.
Thus the behavior of the system on $\bbn$, conditioned on the occurrence
of $F$, is independent of the system to the left of the origin and so is
equivalent to the dynamics of a semi-infinite system on $\bbn$, with a
boundary condition given by an extra site at 0 which is always empty.  It
is this semi-infinite system that we study here, and in fact, for this
system, our arguments apply at all densities.

In this appendix only we write $X=\{0,1\}^\bbn$, define $\tau:X\to X$ to
be the left shift operator, $\tau(x\eta)=\eta$ for $x=0,1$, and say that
a measure $\mu$ on $X$ is {\it $\tau$-invariant} if
$\mu(\tau^{-1}(A))=\mu(A)$ for any measurable $A\subset X$; we define the
density for such a measure to be $\rho_\mu=\mu(\eta(i))$ for
any $i\in\bbn$.  As usual we let $\eta_t$ denote the configuration at
time $t$, under the F-TASEP evolution with boundary condition as
described above, when the initial configuration is $\eta_0$.

\begin{theorem}\label{sinfmain}If $\mu$ is a $\tau$-invariant measure on
$X$ and $n\in\bbn$ is odd then for all $t\ge0$,
$\mu(\eta_t(n))=\rho_\mu$.  \end{theorem}
  
\noindent Note that \cthm{sinfmain} generalizes \cthm{twopt}(a) in two
ways: it is valid for an arbitrary $\tau$-invariant initial measure, and
the result holds at all times, not just in the final state, i.e., not
just for $\eta_\infty$.  By taking $\mu$ to be the Bernoulli measure
$\murho$ and considering the $t\to\infty$ limit we obtain a new proof of
the earlier result.

We begin by introducing two distinct ``coarse grainings''
$\pi_1,\pi_2:X\to \{0,1,d\}^\bbn$. For the first, $\pi_1(\eta)(i)=x$ if
$\eta(2i-1)=\eta(2i)=x$ (where $x=0,1)$ and $\pi_1(\eta)(i)=d$ if
$\eta(2i-1)\ne\eta(2i)$; for the second, $\pi_2(x\eta)=\pi_1(\eta)$ for
$x=0,1$ (here the symbol $d$ stands for ``different'').

\begin{lemma}\label{equal}Suppose that $\eta_0\in X$ and for
$x\in\{0,1\}$ let $\zeta_0=x\eta_0$.  Then for any $t\ge0$,
$\pi_1(\eta_t)=\pi_2(\zeta_t)$.  \end{lemma}

\begin{proof}If $\zeta_t=x\eta_t$ for all $t$, which certainly holds if
$x=0$, then the result is immediate.  We consider then $x=1$ and suppose
that there is a time $t_*$, which we take to be minimal, such that
$\zeta_{t_*}\ne1\eta_{t_*}$.  We will show that then for all $n\ge1$,
$\pi_1(\eta_t)(n)=\pi_2(\zeta_t)(n)$ for all $t\ge0$ and all
$\eta_0\in X$.  The case $n=1$ is easily verified; we proceed by
induction, assuming that the result is true for $n$.  Now necessarily
$\eta_{t_*-1}\rng12=1\,0$, $\zeta_{t_*-1}\rng13=1\,1\,0$, and
$\zeta_{t_*}\rng23=0\,1$, $\zeta_{t_*}(i)=(1\eta_{t_*})(i)$ for $i\ge4$.
Writing $\hat\eta_0:=\tau^2\eta_{t_*}$ and
$\hat\zeta_0:=\tau^2\zeta_{t_*}$ we thus have that
$\hat\zeta_0=1\hat\eta_0$.  Since $\tau^2\eta_{t_*+t}=\hat\eta(t)$ and
$\tau^2\zeta_{t_*+t}=\hat\zeta(t)$, it follows from the induction
hypothesis that $\pi_1(\eta_t)(n+1)=\pi_2(\zeta_t)(n+1)$.  \end{proof}

\begin{proof}[Proof of \cthm{sinfmain}]The result is immediate for
$n=1$.  Now observe that for $x=0,1$ and $n$ odd,
 \be\label{transeq}
\mu\bigl(\eta_t\rng{n}{n+1}=x\,x\bigr)
   =\mu\bigl(\eta_t\rng{n+1}{n+2}=x\,x\bigr).
 \ee
 For if $\eta_0\in X$ and $\zeta_0=y\eta_0$ then from \clem{equal},
$\eta_t\rng{n}{n+1}=x\,x$ if and only if $\zeta_t\rng{n+1}{n+2}=x\,x$,
and \eqref{transeq} follows from the $\tau$-invariance of $\mu$.  But
\eqref{transeq} implies that the distribution of
$(\eta_t(2n-1),\eta_t(2n),\eta_t(2n+1))$ is symmetric under the exchange
of the first and last variables.  From this, and the $n=1$ result the
general case follows by induction.\end{proof}


\begin{thebibliography}{99}


\bibitem{agls}A. Ayyer, S. Goldstein, J. L. Lebowitz, and E. R. Speer,
Limiting States of the Facilitated Partially Asymmetric Exclusion
Process.  In preparation.

\bibitem{bbcs} J. Baik, G. Barraquand, I. Corwin, and T. Suidan.
Facilitated Exclusion Process. Proceedings of the 2016 Abel Symposium.

\bibitem{BM} Urna Basu and P. K. Mohanty, Active-absorbing-state phase
  transition beyond directed percolation: A class of exactly solvable
  models.  {\it Phys. Rev. E} {\bf 79},041143 (2009).

\bibitem{BF} Vladimir Belitsky and Pablo A. Ferrari, Invariant Measures
and Convergence Properties for Cellular Automaton 184 and Related
Processes.  {\it J. Stat. Phys.} {\bf 118}, 589--623 (2005).

\bibitem{BESS} Oriane Blondel, Cl\'ement Erignoux, Makiko Sasadac, and
Marielle Simon, Hydrodynamic limit for a Facilitated Exclusion Process.
{\it Annales de l’Institut Henri Poincar\'e - Probabilités et
Statistiques}, {\bf 56} 667714 (2020).

\bibitem{CZ} Dayne Chen and Linjie Zhao, The Limiting Behavior of the
  FTASEP with Product Bernoulli Initial Distribution. arXiv:1801.10612v1
  [math PR].

\bibitem{Evans} M R Evans, Exact steady states of disordered hopping particle
models with parallel and ordered sequential
dynamics. {\it J. Phys. A: Math. Gen} {\bf 30}, 5669--5685 (1997).

\bibitem{gkr}Alan Gabel, P. L. Krapivsky, and S. Redner, Facilitated
  Asymmetric Exclusion.  {\it Phys. Rev. Lett.}  {\bf 105}, 210603
  (2010).

\bibitem{glsshort}S. Goldstein, J. L. Lebowitz, and E. R. Speer, Exact
solution of the F-TASEP.  J. Stat. Mech. 123202 (2019).

\bibitem{GG}Lawrence Gray and David Griffeath, The Ergodic Theory of
  Traffic Jams. {\it J. Stat. Phys.} {\bf 105}, 413--452 (2001).

\bibitem{hl}Daniel Hexner and Dov Levine,  Hyperuniformity of
Critical Absorbing States. {\it Phys. Rev. Lett.} {\bf 114}, 110602
(2015).

\bibitem{LZGM}E. Levine, G. Ziv, L. Gray, and D.Mukamel, Phase
Transitions in Traffic Models. {\it J. Stat. Phys.} {\bf 117}, 819--830
(2004).

\bibitem{Liggett}Liggett, Thomas M., {\it Interacting Particle Systems.}
Springer-Verlag, New York, 1985.

\bibitem{mcl}Stefano Martiniani, Paul M. Chaikin, and Dov Levine,
Quantifying Hidden Order out of Equilibrium. {\it Phys. Rev. X} {\bf 9},
011031 (2019).

\bibitem{MPZ} Mlotkowski, Wojciech, Penson, Karol A., and ˙\.Zyczkowski,
Karol, Densities of the Raney Distributions.  {\it Documenta Mathematica}
{\bf18} (2013), 1573--1596.

\bibitem{oliveira}M\'ario J. Oliveira,  Conserved Lattice Gas Model with
  Infinitely Many Absorbing States in One Dimension.  {\it Phys. Rev. E} 
  {\bf 71}, 016112 (2005).

\bibitem{Catalan} Steven Roman, {\it An Introduction to Catalan
Numbers}.  Birkh\"auser, New York, 2015.

\bibitem{rpv}Michela Rossi, Romualdo Pastor-Satorras, and Alessandro
Vespignani, Universality Class of Absorbing Phase Transitions with a
Conserved Field. {\it Phys. Rev.  Lett.} {\bf85}, 1803 (2000).

\bibitem{ST}Vladas Sidoravicius and Augusto Teixeira, Absorbing-state
transition for Stochastic Sandpiles and Activated Random Walks. {\it
Electron.  J. Probab.} {\bf 22}, no. 33, 1–35 (2017).

\bibitem{OEIS}N. J. A. Sloane, editor, The On-Line Encyclopedia of
Integer Sequences, published electronically at https://oeis.org (2020).
\end{thebibliography}
\end{document}